\documentclass{article}

\usepackage{amsmath}
\usepackage{amsthm}
\usepackage{proof}
\usepackage{graphicx}
\usepackage{xspace}
\usepackage{color} 
\usepackage[all]{xy}
\usepackage{tikz-cd}
\usepackage{hyperref}
\usepackage{mathtools}
\usepackage{macros}
\usepackage[shortlabels]{enumitem}
\usepackage{libertine}

\begin{document}

\title{The category of \mso transductions}
\author{Miko{\l}aj Boja\'nczyk}
\maketitle 
\begin{abstract}
    \mso transductions are binary relations between structures which are defined using monadic second-order logic.  \mso transductions form a category, since they are closed under composition. We show that many notions from language theory, such as recognizability or tree decompositions, can be defined in an abstract way that only refers to \mso transductions and their compositions.
\end{abstract}



\section{Introduction}

This paper is about the connection between logic and automata, which is one of the central subjects of logic in computer science and formal language theory. The original result, due to B\"uchi, Elgot and Trakhtenbrot (see~\cite[Theorem 3.1]{Thomas97} in the survey of Thomas), says that for finite words, the languages recognized by finite automata (or equivalently, semigroups) are the same as those definable in monadic second-order logic (\mso). This result has seen countless extensions, to objects such as trees or graphs, and also for infinite objects; see~\cite{Thomas97, bojanczyk_recobook} for surveys. In some cases, the recognisability side is best modeled by automata (e.g.~nondeterministic automata for infinite trees, as surveyed in~\cite[Section 6]{Thomas97}). In other cases it is best described by algebras (e.g.~Courcelle's algebras for graphs~\cite{courcelleGraphStructureMonadic2012}, or generalizations of semigroups for countable infinite words~\cite[Section 3]{cartonRegularLanguagesWords2011}). In some cases, both automata and algebras are useful, e.g.~finite words and trees.

Since the relationship between recognisability and definability (in \mso) can be studied in so many settings, and since these studies share many techniques, it is natural to search for common generalizations. An early example of this kind is a paper of Eilenberg and Wright~\cite{eilenbergAutomataGeneralAlgebras}, which studies automata in general algebras, i.e.~algebras that do not necessarily describe finite words. The underlying definition is Lawvere theories, which are the same as finitary monads. Although the early results of Eilenberg and Wright were about free theories, which are the same as automata on ranked trees, the abstract framework can be further developed to handle other kinds of objects, as witnessed by variety theorems for general algebras that were proposed in~\cite{steinbyGeneralVarietiesTree1998,bojanczyk_recobook,chen2016profinite}.

The abstract frameworks discussed above were based on recognizability, mainly in terms of algebras, and \mso plays a minor role in them. An exception is~\cite{monadicMonadic}, where an attempt is made to relate recognisability with \mso definability over abstract structures (formalized using monads). In that work, the point of departure is recognizability, and the connection with \mso is obtained by simulating logical operations (such as Boolean combinations or set quantification) using operations on algebras (such as product or powerset).

In this paper, we take the opposite approach. Our point of departure is \mso logic, and we connect it to recognizability by framing recognizability in terms of logic. For many kinds of structures, the logical aspect is simpler than the algebra aspect, hence the usefulness of reducing algebra to logic in a uniform way. An example of this phenomenon is graphs. It is clear how \mso logic can be applied to graphs (although there are certain design decisions to be made, e.g.~whether or not the set quantifiers can range over sets of edges). In contrast, the algebraic approach to graphs, namely Courcelle's \hr~and \vr~algebras~\cite[Section 2.3 and 2.5]{courcelleGraphStructureMonadic2012}, is arguably less natural and requires a non-trivial amount of book-keeping, such as infinitely many sorts and extra annotation with distinguished vertices or colours. A similar situation arises for matroids, where \mso definability is a clear notion, while a definition of recognizability requires a careful analysis of how a matroid can be parsed by an automaton, see~\cite[Section 3]{hlineny06}.

Our approach of reducing recognizability to logic is based on \mso transductions~\cite{arnborgLagergrenSeese1988, courcelle1991, engelfriet1991}. These are binary relations between structures. The structures could be strings, trees, graphs, etc. The inputs and outputs need not have the same type, e.g.~a string-to-graph transduction inputs a string and outputs a graph, with the output graph being not necessarily unique. Two key properties of \mso transductions are: (a) they are closed under composition; and (b) in the case of languages, i.e.~\mso transductions with Boolean outputs, they coincide with \mso definable languages. Thanks to (a), the \mso transductions form a category (objects are classes of structures, and morphisms are \mso transductions), and thanks to (b), this category describes \mso definable properties. 

As we argue in this paper, many notions from language theory can be described in terms of this category. In the rest of this introduction, we present four examples of such notions. We believe that the corresponding definitions are the main conceptual contribution of the paper. The paper also has some technical results, some of them non-trivial, which are meant to justify how the proposed framework connects to existing results. We would like to underlyine that this paper is about using \mso transductions to uniformly described notions from language theory. It seems to be a pleasing fact that \mso transductions form a category, but the properties of this category (in terms of category theory) are not the subject of this paper. We only remark that the category of \mso transductions seems to share many properties with the category of relations, since \mso transductions are nondeterministic; in particular if we restrict the category of \mso transductions to classes of structures that are finite (i.e.~contain finitely many structures up to isomorphism on structures), then the resulting category is equivalent to the category of relations on finite sets.

Our point of departure, and first definition, is based on a result of Courcelle and Engelfriet~\cite[Theorem 2.6]{courcelle1995logical}, which characterizes treewidth and cliquewidth purely in terms of \mso transductions. This characterization says that a class of graphs has bounded treewidth if and only if it is contained in the set of outputs of some tree-to-graph \mso transduction. For cliquewidth, the characterization is the same, only with a different representation of graphs as logical structures (for treewidth, the universe of the structure is both the vertices and edges, while for cliquewidth, the universe is the vertices only). Motivated by this result, we propose the following notion of width.
\begin{description}
 \item[Bounded width (Definition~\ref{def:bounded-mso-width}).] A class of structures has \emph{bounded \mso width} if it is contained in the image of an \mso transduction that inputs trees.
\end{description}
 The characterizations of Courcelle and Engelfriet show that this notion of width captures treewidth and cliquewidth for graphs; we show in this paper that it also captures branchwidth for matroids representable over finite fields. The essential idea behind the definition is that we do not fix any pre-defined semantics for a tree decomposition, but allow instead any semantics that can be formalized by an \mso transduction (as is the case for known kinds of tree decompositions). This level of generality greatly simplifies notation, while retaining all applications of tree decompositions that are related to \mso.

 The logical definition of bounded width described above is the obvious generalization of the characterizations of Courcelle and Engelfriet~\cite{courcelle1995logical}, and not a new idea on its own. However, as we explain below, the transduction approach can also capture other concepts, such as recognizability or having definable tree decompositions. 
 
 We begin with recognizability. The standard approach to recognizability for a class of structures (such as strings, trees or graphs) is to define an algebra, in which structures from the class (often extended with some kind of annotation, such as distinguished elements) are put together using pre-defined operations (e.g.~strings can be concatenated, or trees can be attached to the root of some other tree). Once the algebra has been defined, a corresponding notion of recognizability arises: a language is called recognizable if it can be recognized by a tree automaton which inputs structures represented using terms built from the pre-defined operations. Similarly to the definition of bounded width, we propose to do away with the pre-defined operations, and we require instead that a tree automaton can be used for any choice of operations that have \mso definable semantics.
\begin{description}
 \item[Recognizability (Definition~\ref{def:recognizability})] If $\Cc$ is a class of structures, a language $L \subseteq \Cc$ is called \emph{recognizable} if for every tree-to-$\Cc$ \mso transduction, the inverse image of $L$ is a recognizable tree language.
\end{description}
One perspective on the above definition is that it reduces recognizability from general classes of structures to the well-understood notion of recognizability for trees.
We show that our notion of recognizability coincides with standard notions of recognizability, using automata or algebras, in all known cases, including strings, trees, traces, graphs and matroids representable over a finite field. (We consider finite structures in this paper, but we believe that the results extend to infinite ones.) 

Since our definition can be applied to any class of structures, it gives a precise mathematical sense to the question: for which classes of structures $\Cc$ are the recognizable subsets are exactly the same as the \mso definable ones? One can easily show that \mso definability always implies recognizability, however, the converse implication may sometimes fail. It is known to fail for graphs of unbounded treewidth or cliquewidth, and we show that it necessarily fails for every class of structure with unbounded width according to our abstract definition. As we discuss below in more detail, we also conjecture that unbounded width is the only possible reason for the conjecture failing, and present some evidence in that direction.

In our investigation of recognizability for general classes of structures, we identify a new kind of reduction, which formalizes intuitive statements such as ``trees are a special case of graphs'', or ``binary relations can be reduced to the special case of symmetric ones''. We call it \mso encodings:

\begin{description}
 \item[Encoding (Definition~\ref{def:mso-encoding}).] Let $\Cc$ and $\Dd$ be classes of structures. A $\Cc$-to-$\Dd$ \mso transduction is called an \emph{\mso encoding} if it admits a one-sided inverse, i.e.~some $\Dd$-to-$\Cc$ \mso transduction such that composing the two gives the identity on $\Cc$.
\end{description}

In the terminology of category theory, the encoding is a section and its one-sided inverse is a retraction. All properties discussed in this paper, are stable under \mso encodings, i.e.~they only need to be proved for the target classes of encodings\footnote{This is not the case for other known orders of classes, e.g.~the order considered by Blumensath and Courcelle~\cite{lmcs:1208} in which $\Cc$ is smaller than $\Dd$ whenever there is a $\Dd$-to-$\Cc$ \mso transduction that is surjective.}. Motivated by this, we investigate the encoding order in more detail. Among other results, we show that binary relations are equivalent, under encodings, to matroids representable over any finite field (and the choice of field is not important), but relations of higher arity are not.

The last of the definitions proposed in this paper concerns definable tree decompositions, which are the standard way of proving that recognizability implies definability. Again, this approach can be described purely in terms of \mso transductions, as follows.
\begin{description}
 \item[Definable tree decompositions (Definition~\ref{def:definable-tree-decompositions}).] A class of structures has \emph{definable tree decompositions} if it admits some \mso encoding into the class of trees.
\end{description}
A simple corollary of the definitions is that if a class of structures has definable tree decompositions, then recognisability implies definability in \mso. We conjecture that bounded \mso width implies definable tree decompositions, and therefore for classes of bounded \mso width, recognizability is equivalent to \mso definability. Since we know that for classes of unbounded \mso width, recognizability is not equivalent to \mso definability, this conjecture would imply a complete characterization of classes where recognisability is equivalent to definability.

The rest of this paper is devoted to a detailed discussion of the above definitions, with examples and proofs of how they capture existing approaches.

\section{\mso transductions}
In this section, we recall the definition of \mso transductions~\cite{arnborgLagergrenSeese1988, courcelle1991, engelfriet1991}. These are transformations between structures such as words, trees, or graphs, which are defined using \mso logic. We assume that the reader is familiar with the basic notions of logic, such as first-order logic, and its extension monadic second-order logic \mso, which can quantify over sets of elements.

 When talking about \mso, we assume that the logic is equipped with modulo counting, which means that for every $p \in \set{2,3,\ldots}$ the logic allows a set predicate ``the number of elements in $X$ is divisible by $p$''. This variant is usually called \emph{counting \mso}, but since this is the only kind of \mso that we use\footnote{We believe that modulo counting is an essential feature of \mso for structures with non-trivial automorphisms, such as graphs. The whole point of \mso is its correspondence with recognizable languages, and for many structures, this correspondence only holds when modulo counting is allowed. An early example of this kind is multisets~\cite[Proposition 6.2]{courcelleMonadicSecondorderLogic1990}. One could argue that \mso without modulo counting is a historical artefact, which arose when the logic was applied to structures such as strings where modulo counting is redundant. According to this perspective, the Seese conjecture has already been solved by Courcelle and Oum~\cite[Theorem 5.6]{courcelle2007vertex}.}, we will simply call it \mso. 

All structures in this paper are finite, in the sense that the universe has finitely many elements. 

\begin{myexample}[Trees]\label{ex:trees}
 Consider trees that are unranked (there is no bound on the number of children), unordered (there is no linear order on the children), and where the nodes are labelled by some alphabet $\Sigma$. Such a tree can be represented as a structure, where the universe is the nodes of the tree, there is a binary ``parent'' relation, and for each label in the alphabet, there is a unary relation that selects nodes with this label.
 Using \mso over this structure, one can describe properties such as ``every node has an even number of children with label $a$''. The expressive power of \mso over such trees is equivalent to tree automata~\cite[Theorem 5.3]{courcelleMonadicSecondorderLogic1990}. When we talk about ``trees'' as a class of structures in this paper, we mean the class of unranked, unordered trees, without labels (which means that the alphabet $\Sigma$ has one letter only).
\end{myexample}

It is clear how an automaton should parse a tree. In contrast, graphs do not have canonical parse trees (save for degenerate cases, such as series parallel graphs), hence the difficulty of defining automata models for graphs. There are no such issues with \mso for graphs. 

\begin{myexample}\label{ex:graphs}[Graphs] 
 When talking about graphs, we mean undirected graphs without self-loops (both choices are non-essential). A graph with vertices $V$ and edges $E$ can be represented as a structure in at least two different ways. In the first representation, which we call \emph{edge representation}, the universe is the vertices, and there is a binary relation for the edges. In the second representation, which we call \emph{incidence representation}, the universe is the vertices and edges, and there is a binary relation for the incidence between a vertex and an edge. 
 For both representations, the vocabulary has one binary relation over elements of the model. When defining graph properties in first-order logic, the two representations are equivalent, because an edge can be represented as a pair of vertices. For \mso, however, the edge representation gives more power. For example, Hamiltonicity can be defined in the incidence representation but not the edge representation~\cite[Proposition 5.13]{courcelleGraphStructureMonadic2012}; the issue is that the incidence representation allows quantification over sets of edges while the edge representation does not. \end{myexample}

We also allow relations that express properties not only of elements, but also of sets. 
\begin{myexample}[Hypergraphs] \label{ex:hypergraphs}
 A \emph{hypergraph} is defined to be a set $V$ together with a family of distinguished subsets of $V$ called \emph{hyperedges}.
 The hyperedges can be seen as a unary relation on sets (and not on elements).
\end{myexample}

\subparagraph*{Vocabularies, structures and transductions.}
Motivated by the above examples, we use structures where relations describe properties of elements and/or sets. To distinguish between arguments of relations that are elements or sets, we use lowercase letters for elements and uppercase letters for sets. For example, a relation $R(x,y,Z)$ has type (element, element, set), because the first two arguments are elements and the last argument is a set. More formally, a \emph{vocabulary} is defined to be a set of relation names, each one with an associated type in $ \set{\text{element, set}}^*$. 
 A \emph{structure} over given a vocabulary consists of a universe, together with an interpretation that assigns to each relation name a relation over the universe of the corresponding type. We require the universe to be finite and nonempty. Finally, a \emph{class of structures} is defined to be any class which contains structures over some common vocabulary, and which is invariant under isomorphism. We do not require the class to be definable in logic. We write $A,B,C$ for structures and $\Cc,\Dd,\Ee$ for classes of structures.

Define a \emph{transduction} to be any binary relation between two classes of structures that is invariant under isomorphism. A transduction can be seen as a nondeterministic operation, where one input may yield multiple outputs. An example of a transduction is the operation that inputs a graph, and outputs the same graph together with a pre-order on its vertices that describes a spanning forest. The central object of this paper is those transductions that can be defined using \mso logic.

\begin{definition}[\mso transduction] \label{def:mso-transduction} An \mso transduction is a transduction that is obtained by composing any finite number of the following \emph{elementary \mso transductions}\footnote{In this definition, an \mso transduction can use any number of elementary \mso transductions, which immediately implies that \mso transductions are closed under composition. One can reduce any \mso transduction to a \emph{normal form}, which uses four elementary \mso transductions: colouring, then copying, then filtering, then interpretation. Some authors define \mso transductions to be those in normal form, and prove closure under composition as a theorem~\cite[Theorem 7.14]{courcelleGraphStructureMonadic2012}.
 } :
 \begin{itemize}
 \item \emph{Interpretation.} Fix some input and output vocabularies. An \mso interpretation is a partial function from structures over the input vocabulary to structures over the output vocabulary that is described by \mso formulas over the input vocabulary in the following way. There is a \emph{universe formula} $\varphi(x)$; the universe of the output structure is defined to be those elements in the input structure that satisfy the universe formula. Next, for every relation name in the output vocabulary, there is a corresponding \emph{relation formula} of the same type (this means that if the relation name has $k$ element arguments and $\ell$ set arguments, then the same is true for the corresponding relation formula), which defines the relation in the output structure. 
 All of these formulas need to be consistent, which means that the relation formulas use only elements that satisfy the universe formula.
 \item \emph{Filtering.} A \mso formula without free variables defines a partial function, which inputs a structure over its vocabulary, outputs the same structure if the formula is true, and is otherwise undefined.
 \item \emph{Copying.} For $k \in \set{2,3,\ldots}$, define $k$-copying to be following function on structures. The output structure is the result of taking $k$ disjoint copies of the input structure and adding a new $k$-ary relation
 which is interpreted as the set of tuples $(a_1,\ldots,a_k)$ that arise by choosing some input element and returning all of its copies. The copies are ordered, i.e.~$a_1$ is the first copy, $a_2$ is the second copy, and so on, which means that the $k$-ary relation is not closed under permuting its arguments. 
 \item \emph{Colouring.} For $k \in \set{1,2,\ldots}$, define $k$-colouring to be the following nondeterministic operation. An output structure is any structure that can be obtained from the input structure by adding $k$ unary relations and interpreting them in any way such that every element satisfies exactly one of the new unary relations. If the input structure has $n$ elements, then there are $k^n$ possible outputs.
 \end{itemize}
\end{definition}

Formally, at the end of an \mso transduction we add an implicit last step, which replaces the output structure with any structure that is isomorphic to it; this guarantees that the \mso transduction is invariant under isomorphism of structures.
This completes the definition of \mso transductions. 

 As defined above, the domain and co-domain of an \mso tranductions are all structures over a given vocabulary. We can restrict them to classes of structures in the following way: if $\Cc$ and $\Dd$ are classes of structures, a $\Cc$-to-$\Dd$ \mso transduction is defined\footnote{This definition ensures that \mso transductions are closed under compositions; this would not be the case if we define a $\Cc$-to-$\Dd$ transduction to be an \mso transduction restricted to pairs from $\Cc \times \Dd$. } to be any \mso transduction from the vocabulary of $\Cc$ to the vocabulary of $\Dd$, which has the property that if the input is from $\Cc$, then all outputs are from $\Dd$. 
We use arrow notation $f : \Cc \to \Dd$ for \mso transductions, even when they are not necessarily functions.


\begin{myexample}[Languages as transductions]\label{ex:language-as-transduction}
 Define $\Bool$ to be a class of structures that has one structure up to isomorphism. The choice of this particular structure is not important, say the vocabulary is empty and the universe has one element. In a $\Cc$-to-$\Bool$ \mso transduction, there are two options for every input: either there are no outputs, or there is one output. Furthermore, an \mso formula can tell which of the two options arises. Therefore, $\Cc$-to-$\Bool$ \mso transductions are the same as an \mso definable languages $L \subseteq \Cc$.
\end{myexample}


\begin{myexample}[Dual]\label{ex:dual}
Consider the \emph{dual} operation, which inputs a graph, and outputs a graph where the vertices are the edges of the input graph, and where two vertices of the output graph are adjacent if the corresponding edges in the input graph share a vertex. This transformation is easily seen to be an \mso transduction under incidence representation. 
 However, it is not an \mso transduction under edge representation, since the universe of the output structure can be quadratic in the universe of the input structure, while \mso transductions increase the universe in an at most linear way.
\end{myexample}

\begin{myexample}[String duplication]
 Although \mso transductions are binary relations, they are not symmetric in the input and output. As an example, consider the binary relation $(w,ww)$ that describes string duplication,
with strings viewed as finite labelled linear orders. This is a partial bijection, but it is an \mso transduction only in one direction, namely from $w$ to $ww$.
\end{myexample}

By definition, \mso transductions are closed under composition. Also, the identity is an \mso transduction. Hence, \mso transductions form a category. 

\begin{definition}
 The category of \mso transductions is the category where the objects are classes of structures, and morphisms between objects $\Cc$ and $\Dd$ are $\Cc$-to-$\Dd$ \mso transductions. 
\end{definition}

In the above definition,
we do not require the classes to be definable in \mso, although this choice seems to have little importance.

We begin with the most basic aspect of this category, namely isomorphism of objects, i.e.~isomorphism between classes of structures. Recall that an \emph{isomorphism} in a category is a morphism which admits a two-sided inverse. (Note that there are two kinds of isomorphism: an isomorphism between two structures, and an isomorphism between two classes of structures.) 
Later in the paper, we will also be interested in one-sided inverses, in fact, \mso transductions with one-sided inverses will be the key definition of this paper, Definition~\ref{def:mso-encoding}. Note that an \mso transduction with a two-sided inverse is necessarily functional, i.e.~every input has exactly one output up to isomorphism.

\begin{myexample}[Edge vs incidence] \label{ex:acyclic-graphs} Consider the two classes of graphs, under edge and incidence representations, respectively, as discussed in Example~\ref{ex:graphs}. These classes are not isomorphic in the category. If there was such an isomorphism, then the dual operation from Example~\ref{ex:dual} would be an \mso transduction also using the edge representation, which it is not. However, if we restrict the classes to acyclic graphs, under the two representations, then they become isomorphic. The nontrivial part is going from the edge representation to the incidence representation, here we use the observation that the edges in an acyclic graph can be represented by its vertices. 
\end{myexample}

As we will argue in Section~\ref{sec:encodings}, it is more useful to consider weaker forms of equivalence than isomorphism. However, even isomorphism can be useful, e.g.~for Cartesian products. The category of \mso transductions has Cartesian products, which is achieved by any natural pairing construction on structures, e.g.~taking the disjoint union of two structures. One advantage of the categorical approach is that one does not need to specify the pairing construction, as long as it is known to exist, since Cartesian products are unique up to isomorphism in the category. The same remarks apply to coproducts.

\section{Tree decompositions}
\label{sec:tree-decompositions}
This section gives a first example of a classical combinatorial notion that can be described purely in terms of morphisms from the category of \mso transductions. This example is not new -- it is based on results of Courcelle and Engelfriet~\cite{courcelle1995logical}. Later in the paper, we give new examples, covering notions such as recognizability, or definable tree decompositions. Readers unfamiliar with treewidth and cliquewidth can simply treat the characterization in Theorem~\ref{thm:courcelle-engelfriet-width-logical} below as the definition of these concepts.
In the theorem, an \mso transduction is called \emph{surjective} if it is surjective as a binary relation, i.e.~every output structure is produced from at least one input structure. (This is a strictly weaker notion than being an epimorphism in the category.) A \emph{subclass of trees} is defined to be any class contained in the class of trees, in the unranked, unordered and unlabelled variant from Example~\ref{ex:trees}. (Other formalizations of trees, e.g.~labelled binary trees, would give the same notion.)

\begin{theorem}[Courcelle and Engelfriet]\label{thm:courcelle-engelfriet-width-logical}
 Let $\Cc$ be a class of graphs. Then 
 \begin{itemize}
 \item \cite[Theorem 2.6]{courcelle1995logical} $\Cc$ has bounded treewidth if and only if there is a surjective \mso transduction from some subclass of trees to the class
 \begin{align*}
 \setbuild{\text{incidence representation of $G$}}{$G \in \Cc$}.
 \end{align*}
 
 \item \cite[Theorem 3.1]{courcelle1995logical} $\Cc$ has bounded cliquewidth if and only if there is a surjective \mso transduction from some subclass of trees to the class
 \begin{align*}
 \setbuild{\text{edge representation of $G$}}{$G \in \Cc$}.
 \end{align*}
 \end{itemize}
\end{theorem}

The theorem motivates the following definition. 

\begin{definition} \label{def:bounded-mso-width} A class of structures $\Cc$ has \emph{bounded \mso width} if there is a surjective \mso transduction from some subclass of trees to $\Cc$. 
\end{definition}


An input tree in the above definition will be called a \emph{tree decomposition} of the output structure; the semantics of the tree decomposition will be the \mso transduction. The above definition is not changed if we require the transduction to be deterministic and total; indeed every surjective \mso transduction from a subclass of trees to $\Cc$ can be made deterministic and total by integrating the colours into the input tree and restricting the input trees. When the \mso transduction is deterministic and total, each tree decomposition represents a unique output structure (although one output structure might have several tree decompositions). A condition equivalent to Definition~\ref{def:bounded-mso-width} is that there is a surjective \mso transduction from the class of all trees to a superset of $\Cc$.

By Theorem~\ref{thm:courcelle-engelfriet-width-logical}, Definition~\ref{def:bounded-mso-width} coincides with bounded treewidth and cliquewidth for graphs, depending on the choice of representation. In Section~\ref{sec:matroids}, we will show that the same is true for representable matroids, with the corresponding notion being branchwidth.

Unlike the usual definitions of cliquewidth and treewidth, Definition~\ref{def:bounded-mso-width} does not define \mso width as a number for every structure. It only tells us when a class of structures has bounded width. From this perspective, there is no difference between width measures that are bounded by functions of each other, such as cliquewidth and hyper-rankwidth~\cite[Theorem 6.3]{oum2006approximating}. However, the notion of \mso width can also be described as a number, using the following variant of hyper-rankwidth. We present the variant for hypergraphs, but the same ideas would apply to other structures. Consider a hypergraph together with a partition of its vertices into two parts. Define the \emph{rank} of this partition to be the rank of the following matrix over the two-element field. The rows of the matrix are subsets of the first part, the columns are subsets of the second part, and the value stored by the matrix in a cell corresponding to row $X$ and a column $Y$ is zero or one, depending on whether $X \cup Y$ is a hyperedge or not. Define a width $k$ tree decomposition of a hypergraph to be a binary tree, where leaves are labeled by vertices of the hypergraph, and such that the rank is at most $k$ for every partition of the vertices that arises by splitting the tree into two parts by removing one tree edge. The \emph{rankwidth} of a hypergraph is the minimal width of its tree decompositions. The following theorem is proved in Appendix~\ref{ap:hyper-rankwidth}.

\begin{theorem}\label{thm:hypergraph-rankwidth}
 A class of hypergraphs has bounded \mso width if and only if it has bounded hyper-rankwidth.
\end{theorem}

We can also define a restricted version of \mso width, which uses strings instead of trees as the decompositions. In the following theorem, a class of strings refers to any subclass of $2^*$, which is the class of strings over an alphabet with two letters. Here, a string is represented as a structure in the standard way, with unary predicates for the letters and a binary predicate for the order on positions. Using a finite alphabet with more than two letters would lead to the same notion.

\begin{definition}
 A class of structures $\Cc$ has \emph{bounded linear \mso width} if there is a surjective \mso transduction from some subclass of strings to $\Cc$. 
\end{definition}

With the same proofs, one can show a linear variant of Theorem~\ref{thm:courcelle-engelfriet-width-logical}, which uses the linear variants of treewidth and cliquewidth, called \emph{pathwidth} and \emph{linear cliquewidth}, respectively.

\section{\mso encodings}
\label{sec:encodings}
Isomorphisms in our category have some usefulness, e.g.~they allow us to ignore the implementation of pairing in Cartesian products, or they enable switching between the incidence and edge representations of acyclic graphs as in Example~\ref{ex:acyclic-graphs}. Other examples include minor differences in representations, such as using the successor instead of the order in a structure representing a string. 
However, an isomorphism can be too much to ask for, and we will have more success with a relaxed notion of equivalence, as described in this section. 

As an example, consider the classes $2^*$ and $4^*$, which are strings over alphabets of sizes two and four, respectively. These are essentially the same class, with the corresponding encoding $4^* \to 2^*$ using two output letters for each input letter. However, this encoding is not an isomorphism, because its image contains only words of even length. One could try to find an improved encoding that is indeed an isomorphism. Even if possible, such an isomorphism would likely be cumbersome. Below, we propose a relaxation of isomorphism, under which $4^*$ and $2^*$ are easily seen to be equivalent, and which is still fine enough so that equivalent classes behave the same way with respect to the properties of interest in this paper. This is one of the key definitions of the paper.

\begin{definition}[\mso encoding]\label{def:mso-encoding}
 An \emph{\mso encoding} of a class $\Cc$ into a class $\Dd$ is a $\Cc$-to-$\Dd$ \mso transduction that admits a one-sided inverse, i.e.~there is a $\Dd$-to-$\Cc$ \mso transduction, called the \emph{decoding}, 
 such that eoncoding followed by decoding is the identity on $\Cc$. 
\end{definition}

Since \mso encodings are easily seen to be closed under composition, the existence of \mso encodings is a pre-order on classes of structures, i.e.~it is a reflexive transitive relation. We call this the \emph{\mso encoding order}. The corresponding equivalence relation (i.e.~encodings both ways) is called \emph{equivalence under \mso encodings}.

\begin{myexample}[Equivalence of strings over various alphabets]\label{ex:equivalence-of-strings-over-nonunary-alphabets}
 The inclusion $2^* \to 4^*$ is an \mso encoding, and the same is true for the \mso transduction of type $4^* \to 2^*$ that uses two output letters for every input letter. Therefore, the two classes $2^*$ and $4^*$ are equivalent under \mso encodings. Similarly, $2^*$ is equivalent to $\Sigma^*$ for every alphabet $\Sigma$ with at least two letters.
\end{myexample} 

The \mso encoding order is related to another order on classes of structures, introduced by 
Blumensath and Courcelle~\cite{lmcs:1208}, which we call \emph{surjective transduction order}. In this order, a class $\Cc$ is below a class $\Dd$ if there is a surjective \mso transduction from $\Dd$ to $\Cc$. Since the decoding in Definition~\ref{def:mso-encoding} must be surjective, it follows that the \mso encoding order refines the surjective transduction order: if a class is below $\Dd$ in the \mso encoding order, then it is also below $\Dd$ in the surjective transduction order. 
 The opposite implication fails, see Example~\ref{ex:two-kinds-of-trees}.

 For some questions about classes of structures, the appropriate order is the surjective transduction order. For example, having a decidable \mso theory, or having bounded \mso width, are notions that downward closed under the surjective transduction order. However, for other notions, the finer ordering of \mso encodings is more appropriate. Two examples of such notions are discussed in Section~\ref{sec:recognizability}: definable tree decompositions, and the equivalence of definability and recognizability.



\begin{myexample}\label{ex:examples-of-encodings}
 The following sequence is strictly increasing under \mso encodings:
 \begin{enumerate}
 \item \emph{Strings.} The class $2^*$, or equivalently, $\Sigma^*$ for any alphabet $\Sigma$ with at least two letters.
 \item \label{it:trees} \emph{Trees.} The class of structures described in Example~\ref{ex:trees}, with or without labels.
 \item \label{it:graphs-incidence} \emph{Graphs under incidence representation,} see Example~\ref{ex:graphs}.
 \item \label{it:graphs-edge} \emph{Graphs under incidence representation,} see Example~\ref{ex:graphs}.
 \item \label{it:k-ary-relations} \emph{$k$-ary relations for $k \ge 3$:} all structures over a vocabulary with one relation name $R(x_1,\ldots,x_k)$. 
 \item \label{it:hypergraphs} \emph{Hypergraphs:} all structures over a vocabulary with one relation name $E(X)$, see Example~\ref{ex:hypergraphs}.
 \end{enumerate}
 For $k=2$, items~\ref{it:graphs-edge} and~\ref{it:k-ary-relations} are equivalent. Furthermore, every class of structures, not only those in the list, admits an \mso encoding to hypergraphs.
 These observations are proved in Appendix~\ref{sec:proof-of-figure}. 
 The encodings are straightforward, while the strictness of the sequence is proved using results from~\cite{lmcs:1208}, and the following growth rate argument. Consider the following table, which counts the number of structures with universe of size at most $n$:

 \begin{center}
 \begin{tabular}{l|l}
 {class} & {log(growth rate)} 
 \\ \hline 
 graphs under incidence representation & $\Theta(n \cdot \log n)$ \\
 $k$-ary relations & $\Theta(n^k)$ \\ 
 hypergraphs & $\Theta(2^n)$
 \end{tabular}
 \end{center}
 Classes that are equivalent under \mso encodings must have the same growth rate (up to linear corrections on $n$), and so the classes from the above table cannot be equivalent. 
\end{myexample}

\begin{myexample}\label{ex:two-kinds-of-trees} In this example, we show that there are two kinds of classes for trees of unbounded depth. All classes in the left column below are equivalent under \mso encodings, likewise for the classes in the right column, but the left column is strictly above the right column.
\begin{center}
 \begin{minipage}[t]{0.49\textwidth}
 \begin{enumerate}
 \item Trees as in Example~\ref{ex:trees}, with or without labels;
 \item Acyclic graphs as in Example~\ref{ex:acyclic-graphs};
 \item Hypergraphs with the following \emph{laminar} condition: every two hyperedges are either disjoint, or one is contained in the other.
 \end{enumerate} 
 \end{minipage}\quad
 \begin{minipage}[t]{0.48\textwidth}
 \begin{enumerate}
 \item Trees as in Example~\ref{ex:trees}, with or without labels, where every node has at most two children;
 \item Connected acyclic graphs as in Example~\ref{ex:acyclic-graphs}, where every vertex degree at most three;
 \item Trees as in Example~\ref{ex:trees}, with or without labels, except that for every node there is a linear order on its children.
 \end{enumerate} 
 \end{minipage}
\end{center}
These equivalences are proved in  Appendix~\ref{sec:proof-of-figure}. The most interesting one, by far, is the equivalence of laminar hypergraphs with trees, and its proof uses modulo counting in a non-trivial way. We use the name ``trees'' for the left column, and ``ordered trees'' or ``binary trees'' for the right column. The two columns are equivalent under the surjective transduction order, and thus this example witnesses that the \mso encoding order is strictly finer.

The classification in this example is not exhaustive, nor is it meant to be. For example, one can bound the depth of trees. Also, for unbounded depth trees there are other possibilities, e.g.~possibly disconnected graphs of degree at most three are strictly between the two columns from the above example. We do not attempt to give a complete description of all classes of structures with respect to \mso encodings. This is in contrast to the coarser hierarchy of Blumensath and Courcelle, where a complete description is conjectured~\cite[Section 9]{lmcs:1208}.
\end{myexample}

\section{Recognisability}
\label{sec:recognizability}
In this section, we phrase recognisability in terms of \mso transductions.

Traditionally, recognizability is defined in terms of algebras. The idea is that for each kind of object (such as words, multisets, trees, graphs, etc. ) one defines an algebraic structure (such as semigroups, commutative semigroups, tree algebras, various graph algebras). Once the algebraic structure has been defined, one calls a language recognizable if it is recognized by a congruence in the algebra that has finitely many equivalence classes. A brief discussion of algebraic recognizability is found in Section~\ref{sec:algebras}, and a more detailed one in~\cite{bojanczyk_recobook}. 

The algebraic notion of recognizability works well for structures such as strings or trees. However, with more complicated structures, such as graphs, algebras become cumbersome, with annotation such as ports, and extra sorts devoted to book-keeping. To avoid such bookkeeping, we propose the following definition. In the definition, a dotted arrow represents a transduction that is not necessarily an \mso transduction, i.e.~it a binary relation closed under isomorphism of structures.

\begin{definition}\label{def:recognizability}
 Let $\Cc$ be a class of structures and let $L \subseteq \Cc$
be a language that is not necessarily definable in \mso. We say that $L$ is \emph{recognizable} if the composition 
\[
 \begin{tikzcd}
 \text{trees} 
 \ar[r,"f"]
 &
 \Cc 
 \ar[r,"L",dotted]
 &
 \Bool
 \end{tikzcd}
 \]
 is an \mso transduction for every $f$ that is an \mso transduction.
\end{definition}

Call an \mso transduction \emph{deterministic} if it does not use colouring; such a transduction is a partial function. 
In the definition above, $f$ ranges over \mso transductions that need not be deterministic. Also the domain is trees which are unordered and not necessarily binary. In Appendix~\ref{sec:ordered-footnote}, we show that these choices are non-essential, for example the same notion would arise if $f$ would range over deterministic total \mso transductions from ordered binary trees.


 

\begin{figure*}
 \begin{tabular}{l|l}
 structures & notion of recognizability that is equivalent\\ & to Definition~\ref{def:recognizability}\\
 \hline
 strings & word automata or, equivalently, semigroups \\
 trees & tree automata or various algebras for trees~\cite{DBLP:books/ems/21/Bojanczyk21} \\ 
 Mazurkiewicz traces & monoid for trance languages~\cite[Section 6.3]{diekert1995book}\\
 graphs of bounded treewidth & \hr-recognisability~\cite[Section 2.3]{engelfrietMSODefinableString2001}\\
 graphs of bounded cliquewidth & \vr-recognisability~\cite[Section 2.5]{engelfrietMSODefinableString2001} \\
 representable matroids of & parse trees~\cite[Section 3]{hlineny06} \\ \quad bounded branchwidth \\
 \end{tabular}
 \caption{\label{fig:other-recognizabilities} Classes of structures for which Definition~\ref{def:recognizability} coincides with previously defined notions.}
\end{figure*}

The above definition is simple, and yet it coincides with the usual algebraic notions of recognizability in all known cases, including objects such as strings, trees, traces, various kinds of graphs, and matroids, see Figure~\ref{fig:other-recognizabilities}. This is discussed in more detail in Section~\ref{sec:algebras}.

An immediate corollary of the definition is that if a language is definable in \mso, then it is recognizable. Indeed, if a language is an \mso transduction (i.e.~a solid arrow, not a dotted one), then pre-composing it with any \mso transduction $f$ will necessarily give an \mso transduction, as a composition of two \mso transductions. For example, this subsumes Courcelle's Theorem~\cite[Theorem 4.4]{courcelleMonadicSecondorderLogic1990}, which says that for graph languages (either in incidence or edge representation), \mso definability implies recognisability\footnote{To be fair, the proof of Courcelle's theorem is simply transferred to showing that Courcelle's notion of recognizability coincides with our notion of recognizability, see Theorem~\ref{thm:algebraic-vs-transducer-recognizability}, which is arguably a better place.}. 

For some classes of structures, recognizability is a strictly weaker notion than \mso definability. The following theorem gives two conditions for such classes, one necessary and one sufficient. Later in this section, we state a conjecture which implies that the theorem gives a dichotomy.

\begin{theorem}\label{thm:when-reco}
 Let $\Cc$ be a class of structures.
 \begin{enumerate}
 \item \label{it:sufficient-non-definable} If $\Cc$ has unbounded \mso width, then some recognizable language $L \subseteq \Cc$ is not definable in \mso. 
 \item \label{it:sufficient-definable} If $\Cc$ admits an \mso encoding in trees, then every recognizable language $L \subseteq \Cc$ is definable in \mso.
 \end{enumerate} \end{theorem}
 \begin{proof}
  The first item is proved  using a diagonalization argument.  Suppose that $\Cc$ has unbounded \mso width, and consider some enumeration $f_1,f_2,\ldots$ of all \mso transductions from trees to $\Cc$. (Such an enumeration is possible, since an \mso transduction has finite syntax.) For every $n$, choose some structure $A_n \in C$ that is not in the image of the first $n$ transductions in the enumeration. Such a structure must exist, because otherwise, we could combine the first $n$ transductions into a single one that is surjective, contradicting the assumption. Every language that contains only structures from the set $\set{A_1,A_2,\ldots}$ 
   is recognizable for trivial reasons, since the image of every \mso transduction from trees to $\Cc$ will contain only finitely many structures from the language, and these structures can be hard-coded into a recognizing formula (the hard-coding is possible thanks to the assumption that all structures under consideration are finite). However, a subset of $\set{A_1,A_2,\ldots}$ can be chosen in uncountably many ways, and thus one of these ways will not be definable in \mso, since there are countably many \mso formulas. To prove that there are uncountably many choices for the subset, we need to observe that the sequence $A_1,A_2,\ldots$ contains infinitely many different structures; this is true since every structure from $\Cc$ is in the image of some \mso transduction.

 For the second item, the proof is in the following diagram.

 
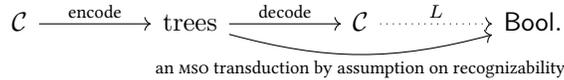
\[
 \begin{tikzcd}
 [column sep=1.5cm]
 \Cc 
 \ar[r,"\text{encode}"]
 &
 \trees 
 \ar[r,"\text{decode}"]
 \ar[rr,"\text{an \mso transduction by assumption on recognizability}"',bend right=15]
 &
 \Cc 
 \ar[r,"L", dotted]
 &
 \Bool.
 \end{tikzcd}
 \]
\end{proof}

 An example of a class as in the first item of the theorem is the class of all graphs, under either edge or incidence representation. Example classes as in the second item of the theorem include the tree classes from Example~\ref{ex:two-kinds-of-trees}. 
 We believe that the assumption in the second item deserves a name:

\begin{definition}[Definable tree decompositions]\label{def:definable-tree-decompositions}
 We say that a class of structures $\Cc$ has \emph{definable tree decompositions} if it admits an \mso encoding in trees.
\end{definition}

The reason for the name is that the image of the encoding in trees can be viewed as an abstract version of a tree decomposition, with the decoding being the semantics of the tree decomposition.

The assumptions in the two items from the above theorem relate $\Cc$ to the class of trees, but using two different orders: the first item uses the surjective transduction ordering of Blumensath and Courcelle, while the second item uses the \mso encoding order. We conjecture that the two conditions are equivalent, despite the orders being different in general.

\begin{conjecture}\label{conjecture:width}
 A class of structures has bounded \mso width iff it has definable tree decompositions. 
\end{conjecture}

The right-to-left implication in the conjecture is immediate: the decoding corresponding to an encoding is a witness of bounded width. It is the left-to-right implication that is the subject of the conjecture. The conjecture is known to be true for every subclass of the class of graphs under incidence representation~\cite[Theorem 2.4]{bojanczykDefinabilityEqualsRecognizability2016a}. For graphs under edge representation, only the special linear case of the conjecture is known: if a class of graphs under edge representation has bounded linear \mso width, then it has definable tree decompositions~\cite[Theorem 3.3]{linearcliquewidth2021}. The conjecture remains open already for graphs under edge representation with bounded (not necessarily linear) \mso width, not to mention $k$-ary relations and hypergraphs. 

 
One can easily see that the conjecture is downward closed under \mso encodings. Therefore, it would be enough to prove that the conjecture is true for all subclasses of the class of hypergraphs, since all classes encode in hypergraphs. 

\begin{proof}[Proof of first three lines in Figure~\ref{fig:other-recognizabilities}]
  Let us now use the notion of definable tree decompositions to prove the first three lines in Figure~\ref{fig:other-recognizabilities}, which say  that our notion of recognizability from Definition~\ref{def:recognizability} coincides with previously defined notions for strings, trees and Mazurkiewicz traces. The remaining lines in the figure, which concern graphs and matroids, will be treated later in the paper, with  Theorems~\ref{thm:algebraic-vs-transducer-recognizability} dealing with graphs, and Appendix~\ref{ap:branchwidth-mso-width} dealing with matroids.

  Let us begin with the case of strings. To prove that the usual notion of recognizability for strings, say in terms of automata, is equivalent to the notion of recognizability from Definition~\ref{def:recognizability}, we use two facts: the usual notion of recognizability for strings coincides with \mso definability, and strings have (trivially) definable tree decompositions. Here is a picture:
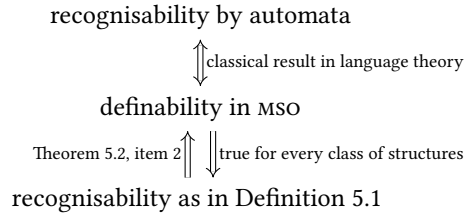
\[
\begin{tikzcd}
\text{recognisability by automata}
\ar[d,Leftrightarrow,"\text{classical result in language theory}"] \\
\text{definability in \mso}
\ar[d,Rightarrow,"\text{true for every class of structures}", shift left=2] \\
\text{recognisability as in Definition~\ref{def:recognizability} }
\ar[u,Rightarrow, shift left=2, "\text{Theorem~\ref{thm:when-reco}, item~\ref{it:sufficient-definable}}"]
\end{tikzcd}
\]
The same argument works for trees, and for Mazurkiewicz traces.
\end{proof}

\subsection{Algebras}
\label{sec:algebras}

In this section, we connect our notion of recognisability from Definition~\ref{def:recognizability} to the algebraic notion of recognisability, as described in Figure~\ref{fig:other-recognizabilities}. The first three rows in the figure -- concerning strings, trees, and traces -- are straightforward. We concentrate on the setting of graphs, with matroids discussed in the next section.

We begin with a more formal discussion of algebras. 
The algebras that we use are multisorted, with possibly infinitely many sorts. 
An \emph{algebra} is defined to be a family of \emph{sorts} (each sort represents a class of objects, e.g.~graphs), together with a set of operations of the form 
\begin{align*}
f : \myunderbrace{A_{1} \times \cdots \times A_{n}}{input sorts} \to \myunderbrace{A}{output sort}. 
\end{align*}
The number $n$ of arguments in an operation can be zero, in which case the operation is called a \emph{constant}. We write $\alga, \algb$ for algebras. A \emph{term} in an algebra is defined to be an ordered tree, with nodes labelled by algebra operations, which is consistent with the sorts in the algebra in the following way: for every node in the term that is labelled by an operation $f$, the number of children is the number of arguments in $f$, and the output sorts of the operations in the children (listed from left to right) match the input sorts of the operation $f$. The \emph{value} of a term is an element of the algebra that is defined in the natural way; this value belongs to the \emph{sort} of the term, which is the output sort of the root operation. 
Terms over an algebra can be seen as labelled ordered trees; if the set of operations in the algebra is finite then the set of labels is finite, and so is the degree of the trees. A \emph{subalgebra} of an algebra is defined to be any algebra that can be obtained from the original algebra by removing some sorts, elements and operations. An algebra is called \emph{finitely generated} if it there is a finite subset of its operations, including constants, such that every element of the algebra is a value of some term using these operations. A finitely generated algebra will have finitely many sorts, namely the output sorts of the generating operations. A \emph{congruence} in an algebra is a family of equivalence relations, one for each sort, which respects the operations of the algebra in the usual sense. A subset of an algebra $\alga$ is called \emph{$\alga$-recognizable} if there exists a congruence which has finitely many equivalence classes on each sort, and such that the subset is a union of equivalence classes of the congruence. 

\subparagraph*{Algebras for treewidth and cliquewidth.}
In his algebraic approach to tree-width and cliquewidth, Courcelle defines two algebras. The first one is called \hr-algebra after hyperedge replacement~\cite[Section 2.3]{courcelleGraphStructureMonadic2012}, and it corresponds to treewidth in the following sense: (i) one of the sorts is the class of graphs under incidence representation; (ii)
 a class of graphs has bounded treewidth if and only if it is contained in some finitely generated subalgebra. Recognizability for this algebra is called \emph{\hr-recognizability}. 
There is another algebra for cliquewidth, called the \emph{\vr-algebra} after vertex replacement~\cite[Section 2.5]{courcelleGraphStructureMonadic2012}. In this algebra, (i) one of the sorts is the class of graphs under edge representation; and (ii) a class of graphs has bounded cliquewidth if and only if it is contained in some finitely generated subalgebra. Recognizability for this algebra is called \emph{\vr-recognizability}. For more details on these algebras, see~\cite{courcelleGraphStructureMonadic2012}.

We now formally relate the algebraic notions recognizability with our notion of recognizability from Definition~\ref{def:recognizability}, explaining the meaning of the two rows in Figure~\ref{fig:other-recognizabilities} that concern graphs. Our notion corresponds to a weakening of the algebraic notion of recognizability: a subset $L$ of an algebra $\alga$ is called \emph{weakly $\alga$-recognizable} if $L \cap \algb$ is $\algb$-recognizable for every finitely generated subalgebra $\algb \subseteq \alga$. Clearly , $\alga$-recognizability implies weak $\alga$-recognisability, and the two notions coincide for finitely generated algebras. For algebras that are not finitely generated, such as the \hr- and \vr-algebras, the two notions are not equivalent, as explained in the following example.

\begin{myexample} \label{ex:weak-weaker} 
 Consider the graph language which contains, for every $n \in \set{1,2,\ldots}$, the graph that is obtained by taking the disjoint union of two copies of the $n \times n$ grid. 
  Since grids have unbounded treewidth, every finitely generated subalgebra of the \hr-algebra will contain finitely many graphs from this language. Therefore, the language is weakly \hr-recognizable. However, the language is not \hr-recognizable, since taking the disjoint union of two graphs is an operation in the \hr-algebra, and thus any congruence in the \hr-algebra would need to remember the size of a square grid in order to recognize the language. The same example applies to the \vr-algebra.
\end{myexample}

For algebras that are not finitely generated, 
it can be debated which version of $\alga$-recognizability, weak or strong, is more useful. For finitely generated algebras, these notions coincide. Since most of algebraic language theory studies languages in finitely generated algebras, such as graph languages with bounded treewidth or cliquewidth, the difference between the two versions seems to be minor.
The following theorem shows that our notion of recognisability corresponds to weak recognizability in the algebras for treewidth and cliquewidth.

\begin{theorem}\label{thm:algebraic-vs-transducer-recognizability}
 For every class of graphs $L$ under incidence representations, recognizability in the sense of Definition~\ref{def:recognizability} is equivalent to weak \hr-recognizability. Likewise for edge representations and weak \vr-recognizability. 
\end{theorem}

\begin{proof}  We use the following lemma.

\begin{lemma}\label{lem:technical-two-recognizabilities}
 Let $\alga$ be an algebra where every sort is a class of structures, such that:
 \begin{enumerate}
 \item \label{assumption:mso-definable} for every sort $\Cc$ and every finite set $\Sigma$ of operations in the algebra, the function of type 
 \begin{align*}
 \text{terms over $\Sigma$ of sort $\Cc$} \to \Cc,
 \end{align*}
 which maps a term to its value, is an \mso transduction;
 \item \label{assumption:factors-through-terms} for every sort $\Cc$, every deterministic \mso transduction 
 \begin{align*}
 f : \text{ordered binary trees} \to \Cc
 \end{align*}
 can be factored through a function as in the previous item, i.e.~there is a finite set of operations $\Sigma$ and an \mso transduction $g$ which makes the following diagram commute 
 
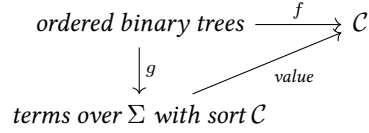
\[
 \begin{tikzcd}
 \text{ordered binary trees} 
 \ar[r,"f"] 
 \ar[d,"g"]
 &
 \Cc \\
 \text{terms over $\Sigma$ with sort $\Cc$} 
 \ar[ur,"\text{value}"']
 \end{tikzcd}
 \]
 
 \end{enumerate}
 Then for every sort $\Cc$ of the algebra, a subset $L \subseteq \Cc$ is recognizable in the sense of Definition~\ref{def:recognizability} if and only if it is weakly $\alga$-recognizable.
\end{lemma}
 \begin{proof}
 The assumptions are chosen so that the proof is a trivial check. 
 
 Let us first show that recognizability implies weak $\alga$-recognizability. Suppose that $L$ is recognizable. To show that $L$ is weakly $\alga$-recognizable, consider some finitely generated algebra $\algb \subseteq \alga$, which is generated by operations $\Sigma$. Consider the following diagram 
 \[
 \begin{tikzcd}
 \text{terms over $\Sigma$ of sort $A$} 
 \ar[r,"\text{value}"]
 &
 \Cc
 \ar[r,"L",dotted]
 &
 \Bool
 \end{tikzcd} 
 \]
 By the assumption that the algebra satisfies condition~\ref{assumption:mso-definable}, the first function is an \mso transduction, and by the assumption on recognizability, the composition of both functions is an \mso transduction. In other words, the set of terms over $\Sigma$ that belong to the language is a regular tree language. Every regular tree language has a corresponding congruence of finite index, which proves that $L \cap \algb$ restricted to the subalgebra is $\algb$-recognizable.

 We now show that weak $\alga$-recognizability implies recognizability. Suppose that $L$ is weakly $\alga$-recognizable. To show that $L$ is recognizable, consider some 
 \mso transduction 
 \begin{align*}
 f : \text{ordered binary trees} \to \Cc.
 \end{align*}
 (We can use ordered binary trees thanks to the results in Section~\ref{sec:ordered-footnote}, and the fact that ordered binary trees and binary trees are equivalent under \mso encodings.)
 The proof is explained in the following diagram.
 
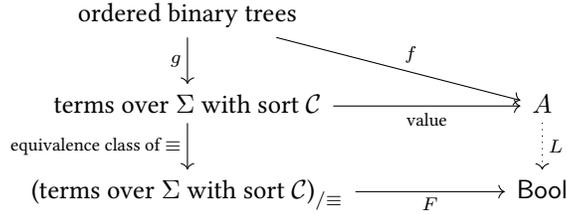
\[
 \begin{tikzcd}
 [column sep=2cm]
 \text{ordered binary trees} 
 \ar[dr,"f"] 
 \ar[d,"g"']
 \\
 \text{terms over $\Sigma$ with sort $\Cc$} 
 \ar[d,"\text{equivalence class of $\equiv$}"']
 \ar[r,"\text{value}"']
 &
 A 
 \ar[d,"L", dotted]
 \\
 \text{(terms over $\Sigma$ with sort $\Cc$)}_{/\equiv} 
 \ar[r,"F"']
 & 
 \Bool
 \end{tikzcd}
 \]

 The upper triangular face in the diagram is the result of applying condition~\ref{assumption:factors-through-terms}. The lower rectangular face is the result of applying the assumption on weak $\alga$-recognisability, which yields a congruence $\equiv$ with finitely many equivalence classes and a function $F$ that decides membership in the language based only on the equivalence class. Since the congruence has finitely many equivalence classes, the equivalence class can be computed by an \mso transduction, and therefore the path down-down-right describes an \mso transduction. 
 \end{proof}

 The algebras used to define recognizability for graphs, namely the \hr-algebra and the \vr-algebra, satisfy the assumptions of the above lemma. Proving this is the essence of the proof of Theorems 2.6 and 3.1 in~\cite{courcelle1995logical}, and is thus not repeated here (the same argument for recognizable matroids is spelled out later in this appendix). Therefore, we can apply the lemma to conclude that Definition~\ref{def:recognizability} coincides with weak recognizability in these algebras, as stated in the theorem. \end{proof}
 
\section{Matroids}
\label{sec:matroids}

We end this paper with an extended example about matroids. 

We begin with matroids that can be represented over a finite field.
A \emph{matroid} represented over a field consists of a set of \emph{elements}, together with a \emph{vector representation}, which is a function from elements to some vector space over the field. A matroid can be described as a structure in two ways:
\begin{itemize}
\item \emph{The independence representation.} The universe is the elements of the matroid, and there is one relation $independent(X)$
which selects the independent subsets of the universe, i.e.~sets whose representing vectors are linearly independent. (If two elements in the set are represented by the same vector, then the set is dependent.)
\item \emph{The null representation.} As in the independence representation, the universe is the elements of the matroid. The independence relation is replaced with a relation on sets, of arity $n = |\field|-1$, which describes the following property
\begin{align*}
\text{(sum of vectors representing $X_1$)} + \cdots + \text{(sum of vectors representing $X_n$)} = 0
\end{align*}
\end{itemize}
The independence representation is the usual notion of matroid when discussing logic over matroids~\cite[Section 2.5]{hlineny06}. One can define the independence relation in terms of the null relation in \mso, since a set is independent if and only if it does not contain any nonempty null sets. When the field has at least three elements, the nulls store strictly more information than the independent sets: one can find two representable matroids with isomorphic independence representations, but non-isomorphic null representations. The case of the two element field (such matroids are called binary) is somewhat special: one can show that the two representations of binary matroids are classes that are isomorphic in the category of \mso transductions. For a binary matroid, the null representation is the same as a hypergraph where hyperedges are closed under symmetric difference.





\begin{theorem}\label{thm:repr-matroids}
 For every finite field $\field$, the class of matroids over $\field$ under null representation is equivalent, under \mso encodings, to the class of binary relations, and admits an encoding into the class of matroids over $\field$ under incidence representation. 
\end{theorem}
We do not know if the encoding of the null representation into the independence representation can be reversed.
Thanks to the above theorem, as long as we use the null representation, the choice of field is unimportant, as long as we use finite fields and we only care about classes modulo equivalence under \mso encodings. 
A corollary of Theorem~\ref{thm:repr-matroids} is that the equivalence of recognizability and definability in \mso, which was proved in~\cite[Theorem 3.5]{linearcliquewidth2021} for classes of binary relations of bounded linear width, entails the analogous result for representable matroids under null representation, see~\cite[Theorem 4.1]{matroids-definability}. However, the techniques used in~\cite{matroids-definability} seem to be more promising; in a future version of this paper, we plan to show that the results from~\cite{matroids-definability} extend to the independence representation.

We now turn to a discussion of branchwidth, which shows that, similarly to tree-width and cliquewidth, it corresponds to our abstract definition of bounded \mso width, as long as representable matroids are concerned. 
We begin by recalling the definition of branchwidth. In a representable matroid, the \emph{rank} of a set of elements is the rank over the vector space spanned by these elements. Tree decomposition are defined in the same way as for hyper-rankwidth for hypergraphs, see Theorem~\ref{thm:hypergraph-rankwidth}, i.e.~by assigning matroid elements to leaves of a binary tree, except the width of a parition of the matroid into two parts $X_1$ and $X_2$ is assigned the following number, which is called \emph{connectivity}:
\begin{align*}
 rank(X_1) + rank(X_2) - rank(X_1 \cup X_2).
\end{align*}
The resulting notion of width for matroids is called \emph{branchwidth}. The following theorem shows that for representable matroids, bounded branchwidth coincides with bounded \mso width.
\begin{theorem}\label{thm:bounded-branchwidth-for-matroids}
 Let $\Cc$ be a subclass of matroids over some field $\field$, under null or independence representation.
 Then $\Cc$ has bounded \mso width if and only if it has bounded branchwidth. 
\end{theorem}

The theorem is proved in Appendix~\ref{ap:branchwidth-mso-width}.
A related concept is linear branchwidth; this is the minimal width of branch decompositions that have the following \emph{linear property}: there is a single root-to-leaf path that contains all nodes that are not leaves. Using the same proof as for Theorem~\ref{thm:bounded-branchwidth-for-matroids}, one can show that a bounded linear branchwidth coincides with bounded linear \mso width.

\subparagraph*{General matroids.}
So far, we have discussed matroids that can be represented using a vector space over some finite field. This is a special case of the general notion of matroids. In general, a matroid~\cite[Section 1.1]{oxley2011matroid} is defined to be a hypergraph (i.e.~a set with a distinguished family of subsets), such that the hyperedges are called \emph{independent sets}, and the following axioms are satisfied: 
\begin{enumerate}
 \item the family of independent sets is nonempty and downward closed under inclusion; and
 \item if $I_1$ and $I_2$ are independent sets, with $|I_1| < |I_2|$, then there is a matroid element $x \in I_2 \setminus I_1$ such that $I_1 \cup \set x$ is an independent set.
\end{enumerate}

The general case of matroids, as defined above, is significantly more general than the representable case. 
For example, the logarithm of the growth rate (see Example~\ref{ex:examples-of-encodings}) for general matroids is $2^{\Theta(n)}$, as proved by Knuth~\cite{knuth1974asymptotic}, which is as large as can be for a class of structures. In particular, a growth rate argument cannot rule out an encoding of hypergraphs in general matroids; we do not know if such an encoding exists. Even for bounded rank, there are many general matroids, as explained in the following theorem.

\begin{theorem}\label{thm:encodings-into-abstract-matroid}
 For every $k$, the class of $k$-ary relations is between the classes of general matroids of ranks $k$ and $2k$, respectively, under the \mso encoding order. 
\end{theorem}
\begin{proof}
In the proof, we pass through the class of \emph{$k$-uniform hypergraphs}, these are hypergraphs where all hyperedges have size exactly $k$. One can easily show that the class of $k$-ary relations is equivalent to the class of $k$-uniform hypergraphs. (For example, an ordered $k$-tuple in a set $X$ can be seen as an unordered set of size $k$ in a set obtained by taking $k$ copies of $X$.) Therefore, to prove the theorem, it will be enough to show that the class of $k$-uniform hypergraphs is between the classes of general matroids of ranks $k$ and $2k$, respectively, as in the following diagram:
\[
\begin{tikzcd}
\txt{general matroids \\ of rank $k$}
\ar[r]
&
\text{$k$-uniform hypergraphs}
\ar[r]
&
\txt{general matroids \\ of rank $2k$}
\end{tikzcd}
\]
The first encoding is simply inclusion, so we concentrate on the second one. This encoding is based on sparse paving matroids that were used by Knuth~\cite{knuth1974asymptotic} to show that the number of general matroids on $n$ elements is doubly exponential in $n$. Consider a $k$-uniform hypergraph $G$. We will define a matroid of rank $2k$, whose vertices are two copies of the vertices in the hypergraph. To define the matroid structure, we will say which sets of size $2k$ are dependent (we use the name \emph{non-basis} for a set of size $2k$ that is dependent). This will specify the matroid structure; since it will tell us which sets of rank $2k$ are independent, and in a matroid of rank $2k$, a set is independent if and only if it is contained in an independent set of size exactly $2k$. The non-bases are defined to be:
\begin{align*}
\setbuild{X_1 \cup X_2}{$X_1$ and $X_2$ are the two copies of some hyperedge in $G$}
\end{align*}
The set of non-bases has the property that one cannot find two non-bases which agree on all but one element; this property implies that the matroid axioms are satisfied by the corresponding independence relation~\cite[Lemma 8]{bansal2015number}.

We now show the decoding. Given a general matroid, the original hypergraph is recovered as follows. We first guess a partition of the matroid elements into two parts, which has the property that every non-basis consists of two subsets of size $k$, one in each part. The original $k$-uniform hypergraph is recovered by returning one of the two parts, and the subsets of size $k$ that can be extended to non-bases.

The fact that $k$ is fixed seems to be important for \mso definability, since we need to define in \mso that the size of a set is $k$, to freely move between bases and non-bases.
\end{proof}
A corollary of the above theorem is that for general matroids, bounded branchwidth (which is defined in the same way as for representable matroids, since general matroids also have a notion of rank) is not equivalent to bounded \mso width, and not particularly useful in the context of \mso and recognisability. Indeed, the class of all binary relations, which has unbounded \mso width, encodes in the class of general matroids of rank at most 4, and thus the latter class has unbounded \mso width, although its branchwidth is trivially 4.

\bibliographystyle{plain}
\bibliography{bib}

\begin{thebibliography}{10}

\bibitem{arnborgLagergrenSeese1988}
Stefan Arnborg, Jens Lagergren, and Detlef Seese.
\newblock Problems easy for tree-decomposable graphs extended abstract.
\newblock In Timo Lepist{\"o} and Arto Salomaa, editors, {\em Automata,
  Languages and Programming}, pages 38--51, Berlin, Heidelberg, 1988. Springer
  Berlin Heidelberg.

\bibitem{bansal2015number}
Nikhil Bansal, Rudi~A Pendavingh, and Jorn~G van~der Pol.
\newblock On the number of matroids.
\newblock {\em Combinatorica}, 35:253--277, 2015.

\bibitem{lmcs:1208}
Achim Blumensath and Bruno Courcelle.
\newblock {On the Monadic Second-Order Transduction Hierarchy}.
\newblock {\em {Logical Methods in Computer Science}}, {Volume 6, Issue 2},
  June 2010.

\bibitem{bojanczyk_recobook}
Miko{\l}aj Boja\'nczyk.
\newblock Languages recognised by finite semigroups, and their generalisations
  to objects such as trees and graphs, with an emphasis on definability in
  monadic second-order logic.
\newblock \url{https://www.mimuw.edu.pl/~bojan/papers/algebra-26-aug-2020.pdf},
  2020.

\bibitem{DBLP:books/ems/21/Bojanczyk21}
Mikolaj Bojanczyk.
\newblock Algebra for trees.
\newblock In Jean{-}{\'{E}}ric Pin, editor, {\em Handbook of Automata Theory},
  pages 801--838. European Mathematical Society Publishing House, Z{\"{u}}rich,
  Switzerland, 2021.

\bibitem{matroids-definability}
Miko{\l}aj Boja{\'{n}}czyk and Colin Geniet.
\newblock Definability equals recognizability for binary matroids of bounded
  linear branchwidth.
\newblock Unpublished manuscript, 2022.

\bibitem{linearcliquewidth2021}
Miko{\l}aj Boja{\'n}czyk, Martin Grohe, and Micha{\l} Pilipczuk.
\newblock {Definable decompositions for graphs of bounded linear cliquewidth}.
\newblock {\em {Logical Methods in Computer Science}}, {Volume 17, Issue 1},
  January 2021.

\bibitem{monadicMonadic}
Miko{\l}aj Boja\'nczyk, Bartek Klin, and Julian Salamanca.
\newblock Monadic monadic second order logic.
\newblock In {\em Samson Abramsky on Logic and Structure in Computer Science
  and Beyond}. Springer, 2023.

\bibitem{bojanczykDefinabilityEqualsRecognizability2016a}
Miko{\l}aj Boja\'nczyk and Michal Pilipczuk.
\newblock Definability equals recognizability for graphs of bounded treewidth.
\newblock In Martin Grohe, Eric Koskinen, and Natarajan Shankar, editors, {\em
  Proceedings of the 31st {{Annual ACM}}/{{IEEE Symposium}} on {{Logic}} in
  {{Computer Science}}, {{LICS}} '16, {{New York}}, {{NY}}, {{USA}}, {{July}}
  5-8, 2016}, pages 407--416. {ACM}, 2016.

\bibitem{cartonRegularLanguagesWords2011}
Olivier Carton, Thomas Colcombet, and Gabriele Puppis.
\newblock Regular {{Languages}} of {{Words}} over {{Countable Linear
  Orderings}}.
\newblock In Luca Aceto, Monika Henzinger, and Ji{\v r}{\'\i} Sgall, editors,
  {\em Automata, {{Languages}} and {{Programming}}}, Lecture {{Notes}} in
  {{Computer Science}}, pages 125--136. {Springer Berlin Heidelberg}, 2011.

\bibitem{chen2016profinite}
Liang-Ting Chen, Ji{\v{r}}{\'\i} Ad{\'a}mek, Stefan Milius, and Henning Urbat.
\newblock Profinite monads, profinite equations, and reiterman’s theorem.
\newblock In {\em International Conference on Foundations of Software Science
  and Computation Structures}, pages 531--547. Springer, 2016.

\bibitem{courcelleMonadicSecondorderLogic1990}
Bruno Courcelle.
\newblock The monadic second-order logic of graphs. {{I}}. {{Recognizable}}
  sets of finite graphs.
\newblock {\em Information and Computation}, 85(1):12--75, March 1990.

\bibitem{courcelle1991}
Bruno Courcelle.
\newblock The monadic second-order logic of graphs v: on closing the gap
  between definability and recognizability.
\newblock {\em Theoretical Computer Science}, 80(2):153 -- 202, 1991.

\bibitem{CourcelleX}
Bruno Courcelle.
\newblock The monadic second-order logic of graphs {X:} linear orderings.
\newblock {\em Theor. Comput. Sci.}, 160(1{\&}2):87--143, 1996.

\bibitem{courcelle1995logical}
Bruno Courcelle and Joost Engelfriet.
\newblock A logical characterization of the sets of hypergraphs defined by
  hyperedge replacement grammars.
\newblock {\em Mathematical Systems Theory}, 28(6):515--552, 1995.

\bibitem{courcelleGraphStructureMonadic2012}
Bruno Courcelle and Joost Engelfriet.
\newblock {\em Graph {{Structure}} and {{Monadic Second}}-{{Order Logic}} - {{A
  Language}}-{{Theoretic Approach}}}, volume 138 of {\em Encyclopedia of
  Mathematics and Its Applications}.
\newblock {Cambridge University Press}, 2012.

\bibitem{courcelle2007vertex}
Bruno Courcelle and Sang-il Oum.
\newblock Vertex-minors, monadic second-order logic, and a conjecture by seese.
\newblock {\em Journal of Combinatorial Theory, Series B}, 97(1):91--126, 2007.

\bibitem{diekert1995book}
Volker Diekert and Grzegorz Rozenberg.
\newblock {\em The book of traces}.
\newblock World scientific, 1995.

\bibitem{eilenbergAutomataGeneralAlgebras}
Samuel Eilenberg and Jesse Wright.
\newblock Automata in {{General Algebras}}.
\newblock {\em Information and Control}, 11:452--470, 1967.

\bibitem{engelfriet1991}
Joost Engelfriet.
\newblock A characterization of context-free nce graph languages by monadic
  second-order logic on trees.
\newblock In Hartmut Ehrig, Hans-J{\"o}rg Kreowski, and Grzegorz Rozenberg,
  editors, {\em Graph Grammars and Their Application to Computer Science},
  pages 311--327, Berlin, Heidelberg, 1991. Springer Berlin Heidelberg.

\bibitem{engelfrietMSODefinableString2001}
Joost Engelfriet and Hendrik~Jan Hoogeboom.
\newblock {{MSO Definable String Transductions}} and {{Two}}-way
  {{Finite}}-state {{Transducers}}.
\newblock {\em ACM Trans. Comput. Logic}, 2(2):216--254, 2001.

\bibitem{hlineny06}
Petr Hlinen{\'{y}}.
\newblock Branch-width, parse trees, and monadic second-order logic for
  matroids.
\newblock {\em J. Comb. Theory, Ser. {B}}, 96(3):325--351, 2006.

\bibitem{knuth1974asymptotic}
Donald~E Knuth.
\newblock The asymptotic number of geometries.
\newblock {\em Journal of Combinatorial Theory, Series A}, 16(3):398--400,
  1974.

\bibitem{oum2006approximating}
Sang-il Oum and Paul Seymour.
\newblock Approximating clique-width and branch-width.
\newblock {\em Journal of Combinatorial Theory, Series B}, 96(4):514--528,
  2006.

\bibitem{oxley2011matroid}
James Oxley.
\newblock {\em Matroid Theory}.
\newblock Oxford Graduate Texts in Mathematics. Oxford University Press, 2011.

\bibitem{steinbyGeneralVarietiesTree1998}
Magnus Steinby.
\newblock General varieties of tree languages.
\newblock {\em Theoretical Computer Science}, 205(1):1--43, September 1998.

\bibitem{Thomas97}
Wolfgang Thomas.
\newblock Languages, automata, and logic.
\newblock In {\em Handbook of formal languages, Vol.\ 3}, pages 389--455.
  Springer, Berlin, 1997.

\end{thebibliography}

\newpage
\appendix

\section{Hyper-rankwidth for hypergraphs}
\label{ap:hyper-rankwidth}
In this part of the appendix, we prove Theorem~\ref{thm:hypergraph-rankwidth}, which says that for classes of hypergraphs, bounded \mso width coincides with bounded hyper-rankwidth.

The definition of hyper-rankwidth uses the rank of a certain matrix to define the complexity of a partition of the vertices in a hypergraph. In this proof, it will be easier to use a different measure, which we call sensitivity.
For a subset of vertices $U$ in a hypergraph, define its \emph{sensitivity} to be the number equivalence classes in the following equivalence relation on subsets of $U$, which considers two subsets $X$ and $Y$ to be equivalent if 
\begin{align*}
X \cup Z \text{ is a hyperedge} \iff Y \cup Z \text{ is a hyperedge} \qquad \text{for every $Z$ disjoint with $U$}.
\end{align*}
This definition is a hypergraph analogue of Myhill-Nerode equivalence.
We now explain why the sensitivity and the rank of the matrix in the definition of hyper-rankwidth are roughly the same quantity. Indeed, the sensitivity of $U$ is equal to the number of distinct rows in the matrix from the definition of hyper-rankwidth. Since the matrix uses a two-element field, it satisfies 
\begin{align}\label{eq:sensitivity-vs-rank}
\text{rank of the matrix} \quad \le \quad 
\myunderbrace{\text{number of distinct rows}}{sensitivity} \quad \le \quad
2^{\text{rank of the matrix}}.
\end{align}
Since the row rank and the column rank are equal, it follows that the sensitivity of a set is bounded an (exponential) function of the sensitivity of its complement. 

Using sensitivity, we prove below the two implications in Theorem~\ref{thm:hypergraph-rankwidth}. 

\subparagraph*{From bounded \mso width to bounded hyper-rankwidth.}
Suppose that a class of hypergraphs $\Cc$ has bounded \mso width. As we remarked after Definition~\ref{def:bounded-mso-width}, the witnessing \mso transduction can use binary trees (instead of general trees) as its domain, and it can also be assumed to be deterministic. Also, we can avoid copying, by adding nodes to the input tree. Summing up, if $\Cc$ has bounded \mso width, then there is an \mso transduction 
\begin{align*}
f : \text{binary trees} \to \text{hypergraphs}
\end{align*}
which does not use copying or colouring, and whose range contains $\Cc$. We can also assume that $f$ does not use filtering, since this will only increase the range, and therefore $f$ is simply an \mso interpretation. In the case of hypergraphs, this means that $f$ is given by two formulas: 
\begin{align*}
\hspace{1cm}\myunderbrace{\varphi(x)}{which nodes of the input tree \\ are vertices of the output hypergraph} 
\hspace{4cm}\myunderbrace{\psi(X)}{which subsets of nodes in the input tree \\ are hyperedges of the output hypergraph} 
\end{align*}
Furthermore, without loss of generality, we can assume that only leaves of the tree satisfy $\varphi(x)$, i.e.~vertices of the output hypergraph are represented using leaves of the input tree.
We will show that the outputs of $f$ have bounded hyper-rankwidth. To prove this, for an input tree $T$, as a tree decomposition of $f(T)$ we will simply use the same tree, since the vertices of the output graph are the leaves in the tree $T$. It remains to justify that this tree decomposition has bounded hyper-rankwidth. By the inequalities from~\eqref{eq:sensitivity-vs-rank}, it is enough to prove that the tree decomposition has bounded sensitivity, which in turn is proved using an Ehrenfeucht-Fra\"iss\'e argument. 

\subparagraph*{From bounded hyper-rankwidth to bounded \mso width.} The proof boils down to showing that a hypergraph can be reconstructed, using \mso, from a bounded hyper-rankwidth decomposition, viewed as a tree labelled by a suitably chosen (finite) set of operation.

Suppose that a hypergraph has hyper-rankwidth at most $k$, as witnessed by some tree decomposition $T$. For a node $x$ of this tree decomposition, define $ U_x$ to be the vertices of the hypergraph that are labels of leaves in the subtree of $x$.  Define $\sim_x$ to be the equivalence relation on subsets of $ U_x$ that appears in the definition of  sensitivity of $ U_x$. Choose for each vertex $x$ some enumeration of the equivalence classes of $\sim_x$.  By the inequalities in~\eqref{eq:sensitivity-vs-rank}, we know that for every subset $ U_x$ has sensitivity is at most $2^k$, and therefore this enumeration is a list of length at most $2^k$.
For a subset of $ U_x$, define its $x$-colour to be the number in $\set{1,\ldots,2^k}$ which represents the index of the corresponding equivalence class in the chosen enumeration. The following claim shows that the colours can be computed in a deterministic bottom up pass.
\begin{claim}\label{claim:compositionality-rankwidth}
 Let $x$ be a node of $T$ with two children $y$ and $z$. There is a function 
 \begin{align*}
 \alpha_x : \set{1,\ldots,2^k} \times \set{1,\ldots,2^k} \to \set{1,\ldots,2^k}
 \end{align*}
 such that for every $Y \subseteq V_y$ and $Z \subseteq V_z$, the $x$-colour of $Y \cup Z$ is obtained by appling $\alpha_x$ to the $y$-colour of $Y$ and the $z$-colour of $Z$. 
\end{claim}
\begin{proof}
 The equivalence class of $Y \cup Z$ under $\sim_x$ is uniquely determined by the equivalence classes of $Y$ and $Z$, under $\sim_y$ and $\sim_z$, respectively.
\end{proof}

Define $S$ to be the tree that is obtained from $T$ by forgetting the original labels, and labelling each node $x$ as follows: (a) if $x$ is a leaf, then it is labelled by the $x$-colours of the two possible sets in $ U_x$, namely $\emptyset$ and $\set x$; (b) if $x$ is not a leaf, then it is labelled by the function $\alpha_x$. This is a tree over a finite alphabet. Using Claim~\ref{claim:compositionality-rankwidth}, there is a finite bottom-up tree automaton with $2^k$ states, which inputs the tree $T$ with a distinguished subset of leaves $X$, corresponding to a subset of vertices in the hypergraph, processes the tree, and computes for each tree node $x$ the $x$-colour of $X \cap  U_x$. Using this automaton, we can write an \mso interpretation which maps the tree $S$ to the original hypergraph. 
\section{\mso encdoings between some classes of structures}
\label{sec:proof-of-figure}
In this part of the appendix, we justify the claims made in Examples~\ref{ex:examples-of-encodings} and~\ref{ex:two-kinds-of-trees} about the existence and non-existence of \mso encodings between selected classes of structures. A summary of all classes of structures discussed in these examples, along with some additional classes (trees of bounded height), can be found in Figure~\ref{fig:encodings}.

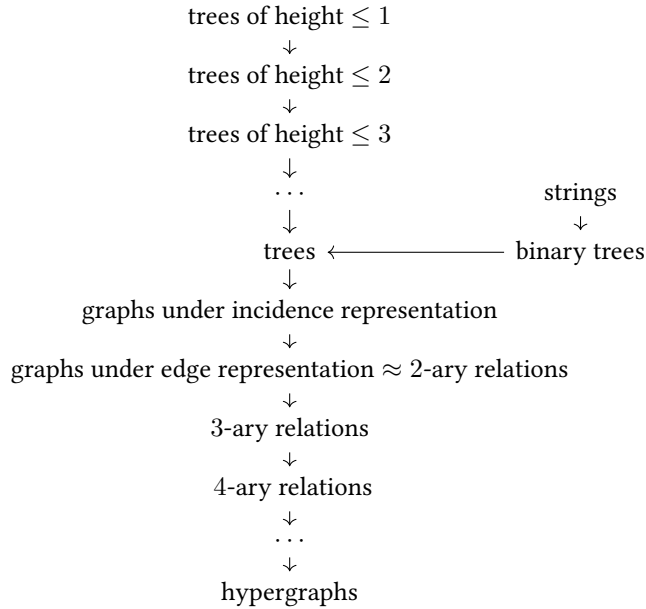
\begin{figure}
 \[
\begin{tikzcd}
 [row sep=0.2cm, column sep=-1cm]
\text{trees of height $\le 1$} 
\ar[d]\\
\text{trees of height $\le 2$} 
\ar[d]\\
\text{trees of height $\le 3$} 
\ar[d]\\
\cdots
\ar[d]
& 
\text{strings}
\ar[d]
\\
\text{trees} 
\ar[d]
&
\text{binary trees}
\ar[l]
\\
\text{graphs under incidence representation} 
\ar[d] \\
\text{graphs under edge representation} \approx \text{$2$-ary relations}
\ar[d] \\
\text{$3$-ary relations} 
\ar[d]\\
\text{ $4$-ary relations} 
\ar[d]\\
\cdots
\ar[d]\\
\text{hypergraphs} 
\end{tikzcd}
\]
\caption{
 \label{fig:encodings} \mso encodings between some classes of structures.
}
\end{figure}

\subsection{Completeness of Figure~\ref{fig:encodings}}
\label{sec:completeness-of-the-figure}
In the proof, 
we discuss separately two parts of the figure, the namely the parts above and below graphs under incidence representation.

\subparagraph*{Upper part.}
We begin with the upper part of the figure, namely:
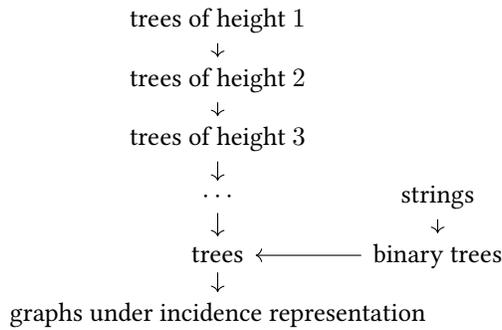
\[
\begin{tikzcd}
 [row sep=0.2cm, column sep=-1cm]
 \text{trees of height $1$} 
 \ar[d]\\
 \text{trees of height $2$} 
 \ar[d]\\
 \text{trees of height $3$} 
 \ar[d]\\
 \cdots
 \ar[d]
 & 
 \text{strings}
 \ar[d]
 \\
 \text{trees} 
 \ar[d]
 &
 \text{binary trees}
 \ar[l]
 \\
 \text{graphs under incidence representation} 
\end{tikzcd}
\]

All of the arrows in the upper part are straightforward encodings (in fact, all encodings are inclusions, except the encoding of strings into binary trees). We focus on their completeness, i.e.~that there are no \mso encodings in the upper part, other than those in the figure and their compositions. We prove this by referring to the surjective transduction order of Blumensath and Courcelle. Since the latter order is strict on the left column in the upper part of the figure, then the left column is also strict in our encoding hierarchy. (Trees and binary trees are equivalent in the hierarchy of Blumensath and Courcelle, but non-equivalent under encodings as we will show below, and therefore the two orders on classes are different.)

It remains to rule out any \mso encodings between the right column, i.e.~strings and binary trees, and the left column discussed above; other than the encoding from binary trees to trees. Blumensath and Courcelle show that for every $k \in \set{1,2,\ldots}$, trees of height $k$ are strictly below strings in their hierarchy, and therefore also in our order. This rules out any encoding from the right column to trees of bounded height. It remains to rule out encodings from the left column to the right column, which is done in the following lemma.

\begin{lemma}
 There is no \mso encoding from trees of height one to binary trees. 
\end{lemma}
\begin{proof}[Proof sketch]
 The issue is that one can use \mso to define a linear order in a binary tree, but not in an unranked tree, even of height one. More formally, we say that a class of structures $\Cc$ has \emph{definable linear orders} if there is an \mso formula 
 \begin{align*}
 \varphi(\myunderbrace{X_1,\ldots,X_n}{set parameters},y,z).
\end{align*}
such that for every structure in $\Cc$, there exists a choice of the parameters such that if this choice is fixed, the formula defines a linear order on all elements in the structure. One can show that if a class $\Cc$ has definable linear orders, then the same is true for every class which encodes in it. Since binary trees have definable linear orders, and trees of height one do not, it follows that there is no \mso encoding from the latter to the former. 
\end{proof}

\subparagraph*{Lower part.}
We now turn to the lower part of the figure: 
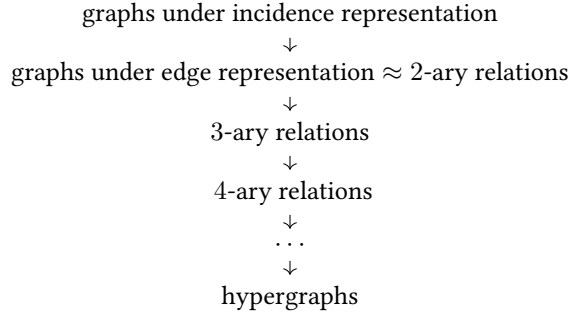
\[
 \begin{tikzcd}
 [row sep=0.2cm, column sep=-1cm]
 \text{graphs under incidence representation} 
 \ar[d] \\
 \text{graphs under edge representation} \approx \text{$2$-ary relations}
 \ar[d] \\
 \text{$3$-ary relations} 
 \ar[d]\\
 \text{ $4$-ary relations} 
 \ar[d]\\
 \cdots
 \ar[d]\\
 \text{hypergraphs} 
 \end{tikzcd}
 \]

Again, the encodings in the figure are straightforward exercises and are left to the reader. (For hypergraphs, we prove in Lemma~\ref{lem:hypergraphs-maximal} that every class of structures admits an encoding to hypergraphs.) We now prove that there are no encodings in the bottom-up direction. 
We prove this using a counting argument. Define the \emph{growth rate} of a class of structures to be the function that maps $n \in \set{1,2,\ldots}$ to the number of non-isomorphic structures in the class that have a universe of size at most $n$. The growth rates of the classes in the lower part of the diagram are given in the following table:

\begin{center}
 \begin{tabular}{l|l}
 {class} & {log(growth rate)} 
 \\ \hline 
 graphs under incidence representation & $\Theta(n \cdot \log n)$ \\
 $k$-ary relations & $\Theta(n^k)$ \\ 
 hypergraphs & $\Theta(2^n)$
 \end{tabular}
\end{center}

The following lemma shows that encodings are monotone with respect to growth rates, up to a linear correction on the size of the structure. Together with the above table of growth rates, the lemma implies that all rows in the lower part of the diagram are non-equivalent. 
 \begin{lemma}\label{lem:counting-argument}
 If there is an \mso encoding from $\Cc$ to $\Dd$, then 
 \begin{align*}
 \exists k \forall n \quad 
 (\text{growth rate of $\Cc$})(n) 
 \ \le \ 
 (\text{growth rate of $\Dd$})(kn).
 \end{align*}
 \end{lemma}
\begin{proof}
 Consider a hypothetical encoding of $\Cc$ in $\Dd$. If we know that the encoding is injective, in the sense that there cannot be an output that arises from two different inputs, then the lemma follows immediately from the linear growth of \mso transductions. It remains to prove injectivity.

 We prove that if there is an encoding, then there is one that satisfies: (*) every output of the encoding is in the domain of the decoding. Once we know (*), then injectivity follows immediately from invertibility. To prove (*), we observe that the domain of the decoding is definable in \mso, like the domain of any \mso transduction. Therefore, to ensure (*), we can use filtering to restrict the outputs of the encoding to this domain. 
\end{proof}

To complete the justification of the claims in Example~\ref{ex:examples-of-encodings}, we prove that every class, not necessarily in the figure, reduces to hypergraphs. 
\begin{lemma}\label{lem:hypergraphs-maximal}
 Every class of structures admits an \mso encoding to the class of hypergraphs.
\end{lemma}
\begin{proof}
 Let $\Cc$ be some class of structures. We will encode structures from $\Cc$ as hypergraphs. We only do the proof for the case when the vocabulary of $\Cc$ has only one relation $R(X_1,\ldots,X_k)$ on sets; the proof for more general vocabularies is done in the same way.
 
 In the first step, we encode $\Cc$ in the class of coloured hypergraphs. Here, a coloured hypergraph is a hypergraph where each vertex is assigned one of $k+1$ colours $\set{0,1,\ldots,k}$. To encode a structure $A \in \Cc$ as a coloured hypergraph, we do the following. Every element $a \in A$ is replaced by $k+1$ copies, with the corresponding colours. Furthermore, there is a hyperedge in the hypergraph which contains all of these $k+1$ copies; this hyperedge will be the only one that will contain colour $0$, and it will allow us to recover which copies correspond to each other. Next, for each tuple of sets $(A_1,\ldots,A_k)$ that is selected by the relation $R$ in the input structure, in the encoding hypergraph we create a hyperedge that contains all first copies of $A_1$, all second copies of $A_2$, and so on up to the $k$-th copies of $A_k$. It is easy to see that this encoding is an \mso transduction, and that there is a corresponding decoding.
 
 Finally, it remains to show how the colours can be encoded, i.e.~that there is an encoding from coloured hypergraphs to (uncoloured) hypergraphs. This is a straightforward construction that is left to the reader. 
 \end{proof}
\subsection{\mso encodings between tree classes}
\label{ex:tree-class-encodings}
We now turn to the tree classes described in Example~\ref{ex:two-kinds-of-trees}.  For convenience, these classes are repeated in Figure~\ref{fig:two-kinds-of-tree-classes}.  In Section~\ref{sec:completeness-of-the-figure}, we have already justified that binary trees are strictly below trees in the \mso encoding order, thus proving that the two columns in Figure~\ref{fig:two-kinds-of-tree-classes} are non-equivalent. We now explain that in each of the two columns, all classes are equivalent under \mso encodings. 

\begin{figure}
    \begin{minipage}[t]{0.49\textwidth}
        \begin{enumerate}
            \item Trees as in Example~\ref{ex:trees}, with or without labels;
            \item Acyclic graphs as in Example~\ref{ex:acyclic-graphs};
            \item Hypergraphs with the following \emph{laminar} condition:  every two hyperedges are disjoint, or one is contained in the other.
        \end{enumerate}                
    \end{minipage}\quad
    \begin{minipage}[t]{0.48\textwidth}
        \begin{enumerate}
            \item Trees as in Example~\ref{ex:trees}, with or without labels, where  every node has at most two children;
            \item Connected acyclic graphs as in Example~\ref{ex:acyclic-graphs}, where every vertex degree  at most three;
            \item Trees as in Example~\ref{ex:trees}, with or without labels, except that for every node there is a linear order on its children.
        \end{enumerate}                
    \end{minipage}
    \caption{\label{fig:two-kinds-of-tree-classes} Two kinds of tree classes from Example~\ref{ex:two-kinds-of-trees}.}
\end{figure}

By far the most interesting equivalence  is the one for laminar hypergraphs. In this equivalence, we use the counting feature of \mso in a non-trivial way. In fact, we believe that counting is essential, and without it, the laminar hypergraphs from the example would no longer be equivalent to trees. This is another argument in favour of having counting as part of the definition of \mso transductions.



\subparagraph*{Left column.}
We begin with the first two items in the left column, namely trees (with or without labels), and acyclic graphs. (Laminar hypergraphs are discussed later.) The equivalence of acyclic graphs and unlabelled trees is left to the reader. (When encoding an acyclic graph in an unlabelled tree, we use nondeterminism to guess the roots.) 
Unlabeled trees are a special case of labeled ones. 
To encode labelled trees in unlabelled trees, we use the encoding that is explained in the following picture: 
    \mypic{1}
    
\subparagraph*{Right column.} We use the name \emph{ordered trees} for trees as in the third item of the right column, i.e.~where every node comes with a linear order on its children.  Labels can be eliminated as in the left column. To encode binary trees in ordered binary trees, an \mso transduction can use nondeterminism to guess the order, by using two colours ``first'' and ``second''. Conversely, ordered binary trees can be encoded in (unordered) binary trees by encoding the sibling order inside the tree structure. 
    Ordered binary trees are a special case of ordered trees. Finally, ordered trees can be encoded in labeled binary trees using the first-child-next-sibling encoding.

\subparagraph*{Laminar hypergraphs}
The rest of this section is devoted to proving that laminar hypergraphs are equivalent to trees, under \mso encodings.  In fact, we prove a slightly stronger result, namely that the class of laminar hypergraphs is isomorphic to the class of trees that satisfy:  (*) if a non-root node has exactly one child, then that child is a leaf. The latter class is easily seen to be equivalent, under \mso encodings, with the class of all trees, by using a straightforward encoding of trees into trees that satisfy (*). 
    
It remains to show the isomorphism between laminar hypergraphs and trees that satisfy (*). This isomorphism is explained in Figure~\ref{fig:laminar}.
\begin{figure*}
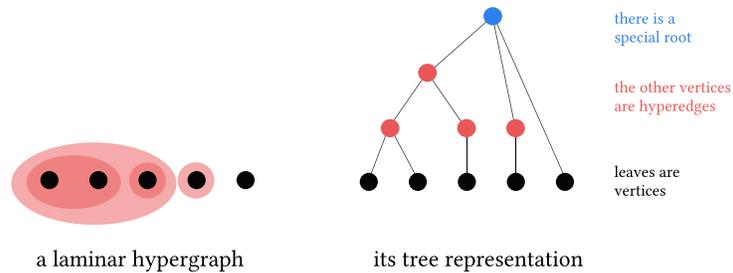

\mypic{5}    
\caption{\label{fig:laminar} A laminar hypergraph and its corresponding tree}
\end{figure*}
In the figure, the tree representation has node colours, but these can be defined even if they are not given,  because the leaves are black, the root is blue, and the remaining nodes are red. The tree representation is easily seen to be a bijection between laminar hypergraphs and trees that satisfy (*).  We will show that both directions of this bijection are in fact \mso transductions, thus proving that the two classes are isomorphic. To go from a tree representation to a laminar hypergraph, we simply use the leaves as the vertices, and as the hyperedges we use the sets that arise by taking some red node and returning its descendants.

The difficult part is producing the tree representation when given a laminar hypergraph. In the proof, we use tree terminology for hyperedges in a laminar hypergraph, such as descendant (smaller inclusion-wise hyperedge), child (maximal inclusion-wise descendant) or leaf (no children).   The difficulty is that the universe of a hypergraph is the vertices and not the hyperedges, and therefore the red nodes of the output tree (which correspond to hyperedges) must be represented using vertices of the input hypergraph. To overcome this difficulty, we will use modulo counting in an essential way.

To represent the hyperedges, we will distinguish between hyperedges that are \emph{branching} (at least two child hyperedges), and the remaining hyperedges which are called \emph{non-branching}.

Non-branching hyperedges can easily be represented in vertices: for each non-branching hyperedge we can choose a vertex that is in it but not in any other smaller hyperedge. If $X$ is the set of  vertices that arise using these choices, then each non-branching hyperedge can be represented by a vertex $x \in X$, namely as minimal inclusion-wise hyperedge that contains $x$. It remains to represent the branching hyperedges. This is solved by reducing branching edges to non-branching ones.

\begin{lemma}
    There exists an \mso formula 
    \begin{align*}
    \varphi(\myunderbrace{X_1,\ldots,X_4}{parameters},Y,Z),
    \end{align*}
    which uses modulo counting,
    such that for every laminar hypergraph, there exists a choice of the parameters $X_1,\ldots,X_4$ such that  for every branching hyperedge $Y$, there is a unique non-branching hyperedge  $Z$  which satisfies the formula.
\end{lemma}

Once we have proved the lemma, we can easily get an \mso transduction that outputs the tree representation of a laminar hypergraph. It remains to prove the lemma. The parameters will be used to choose left and right children for each branching node, i.e.~to choose for each branching node two  distinct children, called the  \emph{left} and \emph{right} children. (If there are more than two children, the all but two of them will be neither left nor right.) Here is an example of such a choice, illustrated in the tree representation:
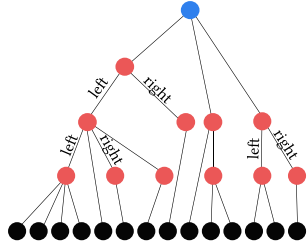
\mypic{7}
Once we have chosen left and right children, we can get a formula as in the lemma using the following procedure: for each branching hyperedge we produce a non-branching one by first going  to the left child, and then taking right children  until a non-branching hyperedge is reached. The procedure can be defined in \mso: the chosen non-branching hyperedge $Z$ is the unique one with the following property: for every intermediate  hyperedge $U$ with $Z \subseteq U \subsetneq Y$ we have
\begin{align*}
\text{$U$ is a child of $Y$}
& \quad \text{iff} \quad 
\text{$U$ is a left child}\\
\text{$U = Y$}
& \quad \text{iff} \quad 
\text{$U$ is non-branching}.
\end{align*}
To complete the proof of the lemma, it remains to show how left and right children can be chosen using a set parameter. 
In the construction, we  use a nondeterministically chosen weight assignment that assigns to each vertex some element of the  cyclic group of order three\footnote{The same construction would work for cyclic groups of order $\ge 3$, and a slightly more involved construction would work for order two. This shows that the \mso encoding of laminar hypergraphs in trees can use modulo counting for any, adversariarily chosen, modulus. }, which we denote by $\mathbb Z_3$. 
Suppose that we have such  a weight assignment (a weight assignment can be described using two set parameters, the namely vertices with weights $0$ and $1$). Define the \emph{weight} of a hyperedge to be the sum -- in the cyclic group -- of the weights of vertices in the hyperedge. This weight can be described in \mso using counting modulo three. By applying the following claim twice, once to get the left children and once to get the right children, we see that any choice of left and right children can be described by two weight assignments, which itself can be described using four set parameters.
\begin{claim}
Consider a laminar hypergraph, and a function that assigns to each  non-branching hyperedge some chosen child. One  can assign weights from $\mathbb Z_3$ to its vertices so that for every non-branching hyperedge, its chosen child has  weight that is different from the other children, and maximal under the order $0 < 1 < 2$.
\end{claim}
\begin{proof}
By induction on the number of hyperedges, we prove a stronger version of the claim: not only can we achieve the desired property in the claim, but we can also ensure that the weight of the entire hypergraph is any group element. The induction basis, when there is only one hyperedge is straightforward. Consider now the induction step. Take a maximal hyperedge $X$ in the hypergraph, and let $X_1,\ldots,X_n$ be its children, with $X_1$ being the chosen child. Suppose that we want a weight assignment as in the claim, and we want the weight of $X$ to be some group element $a$.  Then apply the induction assumption to the children, using the weight assignments that are described in the following table:

\medskip
\begin{tabular}{l|l}
    $a$ & weights for $X_1,\ldots,X_n$ \\
    \hline $0$ & $2,1,0,0, \ldots, 0$ \\
    $1$ & $1,0,0,0, \ldots, 0$ \\
    $2$ & $2,0,0,0, \ldots, 0$ 
\end{tabular}
\medskip

\end{proof}

This completes the proof of the claim and the lemma, and therefore also the encoding of laminar hypergraphs in trees. We do not know if there is an encoding that does not use modulo counting, in fact we conjecture that there is no such encoding.

\section{Alternative definitions of recognizability}
\label{sec:ordered-footnote}
In this part of the appendix, we justify the comments about the definition of recognizability, which say that various constraints on the function $f$ in Definition~\ref{def:recognizability} would give the same notion. To compare these constraints, consider the following definition.

\newcommand{\transclass}{\mathscr F}
\begin{definition}\label{def:d-recognisability}
 Let $\Dd$ be a class of structures. We say that a language $L \subseteq \Cc$, not necessarily definable in \mso, is $\Dd$-recognizable if the composition 
 \[
 \begin{tikzcd}
 \Dd
 \ar[r,"f"]
 &
 \Cc 
 \ar[r,"L",dotted]
 &
 \Bool
 \end{tikzcd}
 \]
 is an \mso transduction for every $\Dd$-to-$\Cc$ \mso transduction. We say that $L$ is \emph{deterministic $\Dd$-recognizable} if the same statement holds, except that $f$ is required to be deterministic (i.e.~not using colouring).
\end{definition}

Using the above definition, we can identify at least four kinds of recognisability, by using either general or deterministic \mso transductions, and using trees or binary trees as the class $\Dd$. These notions are summarised in the following table, with the first notion being the one from 
 Definition~\ref{def:recognizability}.

 \medskip
 
 \noindent \begin{tabular}{c|c|c}
    name & recognisability from Definition~\ref{def:d-recognisability} & class $\Dd$ 
    \\
    \hline  recognizability & $\Dd$-recognizability & trees \\
    det. recognizability & det. $\Dd$-recognizability &  trees \\
    binary recognizability & $\Dd$-recognizability  &  binary trees \\
    det. binary recognizability & det. $\Dd$-recognizability  & binary trees 
 \end{tabular}

 \medskip 
 
 We now justify that all four notions in the table are equivalent.
 The proof has six equivalences: 
 
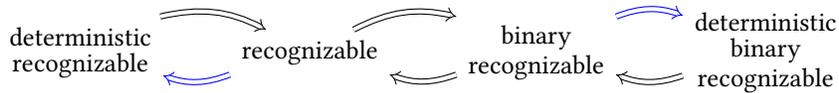
\[
 \begin{tikzcd}
    [column sep=0.9cm]
 \txt{deterministic\\ recognizable }
 \ar[r,Rightarrow,bend left=20] 
 &
 \txt{recognizable}
 \ar[l,blue, Rightarrow,bend left=20]
 \ar[r,Rightarrow,bend left=20]
 &
 \txt{binary \\ recognizable }
 \ar[l,Rightarrow,bend left=20] 
 \ar[r,blue, Rightarrow,bend left=20]
 &
 \txt{deterministic\\  binary \\ recognizable }
 \ar[l,Rightarrow,bend left=20] 
 \end{tikzcd} 
 \]

 The blue implications are immediate, since there are fewer deterministic \mso transductions than general ones. The converses of the blue implications are shown below  using the fact that colours can be encoded in input trees. The middle left-to-right implication will be a straightforward consequence of the fact that there are fewer binary trees than general trees. The most interesting implication will be the middle right-to-left implication, from binary recognizable to recognizable. The detailed arguments are presented below.

 \subparagraph*{Deterministic recognizable implies recognizable.} Suppose that $L$ is deterministic recognizable, and consider some -- possibly nondeterministic -- \mso transduction $f$. Like any \mso transduction, $f$ is the composition of a colouring transduction followed by a deterministic \mso transduction. By encoding the colours in the tree structure of an unlabelled tree, we can decompose $f$ as a composition 
 \[
 \begin{tikzcd}
 \text{trees} 
 \ar[r,"f_1"]
 &
 \text{trees}
 \ar[r,"f_2"] 
 & 
 \Cc
 \end{tikzcd}
 \]
 where $f_2$ is deterministic. By the assumption that $L$ is deterministic recognizable, we know that $f_2;L$ is recognizable. Therefore, also 
 \begin{align*}
 \myunderbrace{f_1;f_2}{$f$};L
 \end{align*}
 is recognizable. The same proof works in the binary case.
 \subparagraph*{Recognizable implies binary recognizable.} Since binary trees admit an \mso encoding to trees, this part of the proof will follow from the following claim. 
 \begin{claim}\label{lem:recognizable-and-mso-reductions}
 If there is an \mso encoding from $\Dd_1$ to $\Dd_2$, then $\Dd_2$-recognizability implies $\Dd_1$-recognisability. Likewise for deterministic recognizability.
 \end{claim}
 \begin{proof}
 Assume that a language $L \subseteq \Cc$ is $\Dd_2$-recognizable. We want to show that is also $\Dd_1$-recognizable, which means that for every \mso transduction of type $\Dd_1 \to \Cc$, its post-composition with $L$ is an \mso transduction. The relevant transductions are shown in the following diagram:
 \[
 \begin{tikzcd}
 \Dd_2 
 \ar[r,bend left=20,"\text{decode}"]
 & \Dd_1
 \ar[l,bend left=20,"\text{encode}"]
 \ar[r,"f"]
 &
 \Cc 
 \ar[r,"L"]
 &
 \set{\text{yes,no}}
 \end{tikzcd}
 \]
 By assumption on $\Dd_2$-recognizability, we know that 
 \begin{align*}
 \text{decode};f;L
 \end{align*}
 is an \mso transduction. Therefore, its pre-composition with the encoding, which is the same as $f;L$, is an \mso transduction.
 \end{proof}

 \subparagraph*{Binary recognizable implies recognizable.} This is the most interesting case. Since binary trees and ordered trees are equivalent under \mso encodings, in light of Claim~\ref{lem:recognizable-and-mso-reductions} it will be enough to show that ordered recognizable implies recognizable, where \emph{ordered recognizable} refers to $\Dd$-recognizability for the class of ordered trees. Suppose that 
 \[
\begin{tikzcd}
\Cc 
\ar[r,"L",dotted]
&
\Bool
\end{tikzcd}
\] 
 is ordered recognizable. Consider some \mso transduction 
 \[
 \begin{tikzcd}
 f : \text{trees} \to \Cc.
 \end{tikzcd}
 \]
 We want to show that $L;f$ is an \mso transduction. Consider the projection 
 \begin{align*}
 \pi : \text{ordered trees} \to \text{trees}
 \end{align*}
 which forgets the order in an ordered tree. By the assumption on ordered recognizability, the composition
\[
\begin{tikzcd}
\text{ordered trees}
\ar[r,"\pi"]
&
\text{trees}
\ar[r,"f"]
&
\Cc
\ar[r,dotted,"L"]
&
\Bool
\end{tikzcd}
\]
 is an \mso transduction, i.e.~an \mso definable language of ordered trees. This language of ordered trees is order invariant, i.e.~changing the sibling order in an ordered tree does not affect membership in the language. Since our variant of \mso has modulo counting built in, we can use a result of Courcelle~\cite[Corollary 4.3]{CourcelleX}, which says that if a language of ordered trees is order invariant and definable in \mso, then it can be defined in \mso without using the sibling order. Therefore, the corresponding language of (unordered) trees is definable in \mso.


\section{Matroids}
In this part of the appendix, we present the missing proofs for the results in Section~\ref{sec:matroids} about matroids. 

\subsection{Proof of Theorem~\ref{thm:repr-matroids}}
In this part of the appendix, we prove Theorem~\ref{thm:repr-matroids}, which says that for every finite field $\field$, the class of matroids over $\field$ under null representation is equivalent, under \mso encodings, to the class of binary relations, and admits an encoding into the class of matroids over $\field$ under incidence representation.

 Fix the finite field $\field$. In the proof, we use two other classes. The class of \emph{bipartite graphs} is the class of binary relations which describe a bipartite graph, together with a unary relation that distinguished the two parts of the bipartite graph. The class of \emph{matrices} (over the fixed field $\field$) is defined as follows: a \emph{matrix} to be two sets, the rows and the columns, and a function which assigns to each pair (row, column) a field element. A matrix is represented as a structure, where the universe is the disjoint union of the rows and columns, and for every field element $a \in \field$ there is a binary relation 
 which selects (row, column) pairs that have label $a$. 
 
 To prove Theorem~\ref{thm:repr-matroids}, we will show that following encodings:
 
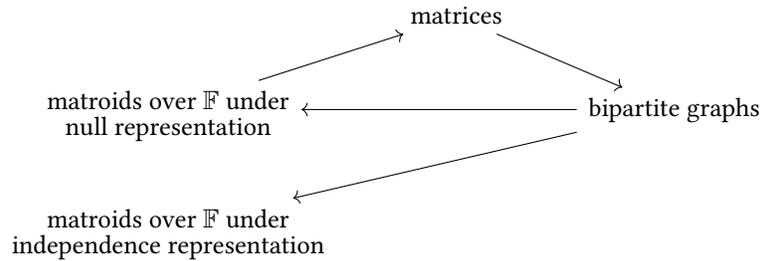
\[
 \begin{tikzcd}
 & \text{matrices} 
 \ar[dr]
 \\
 \txt{matroids over $\field$ under \\ null representation}
 \ar[ur]
 & 
 &
 \text{bipartite graphs}
 \ar[ll]
 \ar[lld]\\
 \txt{matroids over $\field$ under \\ independence representation}
 \end{tikzcd}
 \]
The two encodings that start in bipartite graphs are essentially the same one, i.e.~the encoding is capable of defining the null structure, but the decoding does not need to use it.
 
 \subparagraph*{From matrices to bipartite graphs.} A matrix is the same thing as a bipartite graph with edges coloured by field elements. We first show how the edge colours can be encoded in vertex colours, as explained in the following picture: 
 \mypic{3}
 Next, one can encode the vertex colours in an uncoloured graph. One way to do this is to use odd cycles of different colours for vertices with different colours; this works because the input graph has only even-length cycles. The details are left to the reader. 

 \subparagraph*{From bipartite graphs to matroids.} We will show how to encode a bipartite graph as a matroid. The encoding works for both the independence and null representation. Consider a bipartite graph, as in the following picture: 
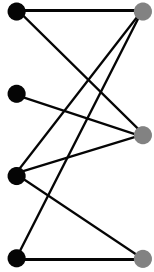
\mypic{4}
The corresponding matroid is defined as follows. We use a matroid with elements coloured by two colours, but we explain later how to eliminate the colours. The elements of the matroid are the same as the vertices in the bipartite graph, with two colours for representing the left and right vertices. The elements that represent the right vertices are a basis of the underlying vector space. For each element that represents a left vertex, the corresponding vector is the sum of the basis vectors that correspond to its neighbours. It is easy to see that this matroid can be produced by an \mso transduction, using either the null or independence representation.

From this matroid, there is a decoding that recovers the original bipartite graph: for an element $x$ that represents a left vertex, its neighbours are recovered by taking the unique minimal dependent set that contains $x$ as the only left vertex. This decoding uses only the independence relation, so it will work if the matroid uses either the independence or null representation.

Finally, we explain how to encode the colours in the matroid. One way to do it is to use duplicate elements: for each element in the right side of the bi-partite graph, we have two copies that represent the same vector. These will be the only duplicate elements in the matroid (two elements are duplicates if they represent the same vector), assuming that the original bi-partite graph satisfied the assumption (*): there are no two elements on the left side of the graph that have equal neighbourhood sets. Bi-partite graphs that satisfy (*) are easily seen to be equivalent to general bi-partite graphs under \mso encodings.

\subparagraph*{From matroids to matrices.} Consider a matroid, given by its null representation. Choose a basis of the matroid, i.e.~a maximal set of independent matroid elements. Based on this basis, we can define a matrix: the rows are matroid elements, the columns are basis elements, and the cells describe the coefficients of each matroid element in its basis decomposition. This matrix, which is not unique because it depends on a choice of basis, can be produced by an \mso transduction, which uses the nondeterministic colouring to choose the basis. This transduction is an encoding, since the null representation can be recovered from the matrix.

 \subsection{Branchwidth as a special case of \mso width }
\label{ap:branchwidth-mso-width}
In this part of the appendix, we prove Theorem~\ref{thm:bounded-branchwidth-for-matroids}, which says that, in the case of representable matroids, our notion of \mso width coincides with the usual width notion for matroids, namely branchwidth. (Recall that this is no longer true for general matroids, as we remark after Theorem~\ref{thm:encodings-into-abstract-matroid}.) 

One way of proving this theorem would be to appeal to Theorem~\ref{thm:hypergraph-rankwidth}, about hyper-rankwidth for hypergraphs, and to show that for classes of representable matroids over a finite field, bounded branchwidth and bounded hyper-rankwidth coincide. (As in the previous remark, these two notions \emph{do not} coincide for general matroids. For general matroids, the appropriate combinatorial notion in the context of this paper is hyper-rankwidth and not branchwidth.)

We choose to present a different proof, which is longer, but it has the advantage that it also shows a second result, namely the version of Theorem~\ref{thm:algebraic-vs-transducer-recognizability} for representable matroids, i.e.~that the notion of recognisability from Definition~\ref{def:recognizability} coincides with the notion of recognizability from the literature~\cite{hlineny06}. This explains the meaning of the last row in Figure~\ref{fig:other-recognizabilities}.

Fix a finite field $\field$. In the proof, we work with the following algebra, inspired by~\cite{hlineny06}.

\begin{definition}
 [Branchwidth algebra] Define a $k$-ported matroid to be a matroid with $k$ distinguished elements. The \emph{branchwidth algebra} has one sort for each $k \in \set{0,1,\ldots}$, with this sort describing $k$-ported matroids. The operations are:
 \begin{enumerate}
 \item \textbf{Constants.} Every matroid with ports is a constant.
 \item \textbf{Renaming ports.} Consider any function 
 \begin{align*}
 \alpha : \set{1,\ldots,\ell} \to \set{1,\ldots,k}
 \end{align*}
 with $\ell, k \in \set{0,1,\ldots}$.
 This function yields an operation of type 
 \begin{align*}
 \text{$k$-ported matroids} \to \text{$\ell$-ported matroids}
 \end{align*}
 defined as follows: the matroid is not changed, and the $i$-th distinguished element in the output is the $\alpha(i)$-th distinguished element in the input.
 \item \textbf{Quotienting.} Consider a sequence of field elements $a_1,\ldots,a_k \in \field$. This sequence yields a operation of type 
 \begin{align*}
 \text{$k$-ported matroids} \to \text{$k$-ported matroids}
 \end{align*}
 defined as follows: the vector space in the matroid is quotiented by identifying with zero the vector obtained by taking the linear combination of distinguished elements with coefficients $a_1,\ldots,a_k$. 
 \item \textbf{Disjoint union.} For every $\ell,k \in \set{0,1,\ldots,}$ there is an operation of type 
 \begin{eqnarray*}
 & \text{($k$-ported matroids)} \times \text{($\ell$-ported matroids)} \to \text{$(k+\ell)$-ported matroids}
 \end{eqnarray*}
 which outputs the disjoint union of the two inputs, defined in the natural way.
 \end{enumerate}
\end{definition}
 
Adjusting for different notation, \cite[Theorem 3.4]{hlineny06} shows that a class of matroids representable over the field has bounded branchwidth if and only if it is contained in some finitely generated subalgebra of the branchwidth algebra. Also, the accepted notion of recognisability for matroids, see~\cite{hlineny06}, is recognisability in the branchwidth algebra. The main part of our proof will be the following lemma.

 \begin{lemma} \label{lem:brancwidth-algebra-is-legit} Let us view the sorts of the branchwidth algebra as classes of structures, using the null representation of matroids. The branchwidth algebra satisfies the two assumptions of Lemma~\ref{lem:technical-two-recognizabilities}. The same is true for the independence representation.
 \end{lemma}

 Before proving the lemma, we observe two consequences. 
 Thanks to the above lemma and Lemma~\ref{lem:technical-two-recognizabilities}, we see that weak recognizability in the branchwidth algebra coincides with our notion of recognizability from Definition~\ref{def:recognizability}; this is the meaning of the last row in the table from Figure~\ref{fig:other-recognizabilities}. Also, the above lemma implies that bounded branchwidth is equivalent to \mso width. As mentioned above, bounded branchwidth is equivalent to being contained in some finitely generated subalgebra of the branchwidth algebra. By assumption~\ref{assumption:mso-definable} in Lemma~\ref{lem:technical-two-recognizabilities}, this is a sufficient condition for bounded \mso width, and by assumption~\ref{assumption:factors-through-terms}, this is a necessary condition. It remains to prove Lemma~\ref{lem:brancwidth-algebra-is-legit}.

 \begin{proof}[Proof of Lemma~\ref{lem:brancwidth-algebra-is-legit}] We begin with the first assumption of Lemma~\ref{lem:technical-two-recognizabilities}, which says that for every sort $\Cc$ of the branchwidth algebra, and every finite set $\Sigma$ of operations in the algebra, the function of type 
 \begin{align*}
 \text{terms over $\Sigma$ of sort $\Cc$} \to \Cc,
 \end{align*}
 which maps a term to its value, is an \mso transduction. The same is true for the independence representation.
 The case of the null representation is proved in a standard way.
 For every $n \in \set{1,2,\ldots}$ one can write a tree automaton which inputs a term together with a $n$ distinguished sets of leaves (constant) and answers whether or not adding the vectors corresponding to these $n$ sets gives a null vector. The case of independence representation follows, since we can always forget the null information and keep only the independence information.

 Consider now the second assumption. This assumption says that for every sort $\Cc$, every deterministic \mso transduction 
 \begin{align*}
 f : \text{ordered binary trees} \to \Cc
 \end{align*}
 can be factored through a function as in the previous item, i.e.~there is a finite set of operations $\Sigma$ and an \mso transduction $g$ which makes the following diagram commute 
 \[
 \begin{tikzcd}
 \text{ordered binary trees} 
 \ar[r,"f"] 
 \ar[d,"g"]
 &
 \Cc \\
 \text{terms over $\Sigma$ with sort $\Cc$} 
 \ar[ur,"\text{value}"']
 \end{tikzcd}
 \]
 The proof for this is the same as in Theorem~\ref{thm:hypergraph-rankwidth}. As in that proof, it is enough to consider the case when $f$ is an \mso interpretation, i.e.~it does not use colouring, copying or filtering. As in that proof, for input tree for $f$, if we partition that tree into two parts across a single edge, then the resulting partition of the elements in the output matroid will have bounded sensitivity, equivalently, bounded hyper-rankwidth. For matroids representable over a fixed field, sensitivity and connectivity are roughly the same measure as stated in the following claim. 

 \begin{claim}\label{claim:connectivity-vs-sensitivity}
 For every finite field $\field$, \begin{align*}
 \text{sensitivity of $X_1 \cup X_2$} \le 
 \text{connectivity of $X_1 \cup X_2$} \le 
 1 + |\field|^{\text{sensitivity of $X_1 \cup X_2$}}.
 \end{align*}
 holds for every partition $X_1 \cup X_2$ of the elements in a matroid representable\footnote{The claim fails in general matroids. For the same reasons as explained after Theorem~\ref{thm:encodings-into-abstract-matroid}, in general matroids of rank at most four one can have partitions of unbounded sensitivity; however all of these partitions will necessarily have connectivity at most four. 
 } over $\field$.
 \end{claim}
 \begin{proof}
 Consider a partition $X_1 \cup X_2$. The connectivity of this partition is equal to the rank of the following subspace of the representing vector space: 
 \begin{equation}
 \label{eq:interface-subspace}
 \bigcap_{i = 1,2} \text{vectors spanned by the elements of $X_i$}
 \end{equation}
 vectors that are spanned by subsets of $X_1$, and also spanned by subsets of $X_2$. The number of different vectors in this subspace is bounded by the sensitivity, and therefore the connectivity is at most the sensitivity. For the converse inequality, we observe that the equivalence class of subset $Y \subseteq X_1$ with respect to the equivalence class in the definition of sensitivity is uniquely determined by two pieces of information:
 \begin{enumerate}
 \item is the sum of vectors representing elements of $Y$ equal to some vector in the subspace~\eqref{eq:interface-subspace};
 \item if yes, then which vector is it equal to.
 \end{enumerate}
 The number of possible values for this information is bounded by the quantity at the end of the inequality in the statement of the claim.
 \end{proof}
 
 By the above claim,  for every binary tree $T$, if we view $T$ as a branch decomposition of the matroid $f(T)$, then this decomposition has branchwidth that is bounded by a constant depending only on $f$. Finally, using the correspondence of branchwidth decompositions with terms in the branchwidth algebra, we can relabel the nodes of $T$ with operations in the branchwidth algebra that construct the matroid $f(T)$. In this relabeling, the invariant is that a subtree of the $T$ generates the matroid obtained from $f(T)$ by restricting it to elements with origin in the subtree. 

 This completes the proof of Lemma~\ref{lem:brancwidth-algebra-is-legit}.
\end{proof}

 


\end{document}